\documentclass[a4paper]{article}

\usepackage[utf8]{luainputenc}

\title{Existential Second-Order Logic Over Graphs: A~Complete Complexity-\!Theoretic Classification}

\author{Till Tantau}
\date{\small
  Institute of Theoretical Computer Science\\Universität zu
  Lübeck, Germany\\\texttt{tantau@tcs.uni-luebeck.de}\\[.5em]
  December 16th, 2014
}

\usepackage{amssymb, amsfonts, amsmath, amsthm}

\theoremstyle{plain}
\newtheorem{theorem}{Theorem}[section]
\newtheorem{lemma}[theorem]{Lemma}

\theoremstyle{definition}

\usepackage[text={140mm,222mm},centering]{geometry}

\usepackage{tikz}

\usetikzlibrary{
  backgrounds,
  positioning,
  shapes,
  decorations.pathmorphing,
}

\tikzset{
  every picture/.style={semithick},%
  node/.style={draw,circle,minimum size=5mm,inner
    sep=0.5pt,font=\footnotesize,fill=white},%
  small node/.style={node,minimum size=3.5pt,inner sep=0pt,outer
    sep=0pt,font=\tiny},%
  on/.style={fill=white,inner sep=-.4pt,circle}
}

\newtheorem{fact}[theorem]{Fact}

\let\normalsmall=\small
\newcommand\specialsmall{\let\normalsmall=\scriptsize}
\newcommand\abstractsmall{\let\normalsmall=\footnotesize}
\newcommand\Class[1]{\mathchoice{\text{\normalfont\normalsmall$\mathrm{#1}$}}{\text{\normalfont\normalsmall$\mathrm{#1}$}}{\text{\normalfont$\mathrm{#1}$}}{\text{\normalfont$\mathrm{#1}$}}}       
\newcommand{\Lang}[1]{\ifmmode{\text{\textsc{#1}}}\else\textsc{#1}\fi}

\newcommand\FaginDef[1]{\Class{FD}_{\mathrm{#1}}}
\newcommand\FaginDefbar[1]{\Class{\overline{FD}}_{\mathrm{#1}}}
\newcommand\Models[1]{\Lang{models}_{\mathrm{#1}}}
\newcommand\hilight[1]{\colorbox{black!10}{\vrule width0pt height1.3ex\smash{#1}}}
\newcommand\patterntwo[8]{%
  \smash{\!%
  \tikz[baseline=-.5ex,any/.style={black!25,on/.style={inner sep=-.4pt,circle}}]{
    \node [node, minimum size=3mm] at (1,0) (w) {};
    \node [node, minimum size=3mm,fill=black] at (0,0) (b) {};
    \draw (w) edge [out=25, in=-25,looseness=12,#1] node  [on,near start,anchor=center] {\scriptsize#2} (w);
    \draw (b) edge [out=155, in=205,looseness=12,#3] node  [on,near
    start, anchor=center] {\scriptsize#4} (b);
    \draw (b) edge [bend left=20,#5] node [on] {\scriptsize#6} (w);
    \draw (w) edge [bend left=20,#7] node [on] {\scriptsize#8} (b);
  }\!}%
}
\newcommand\patterntwomixed[4]{%
  \smash{%
  \tikz[baseline=-.5ex,any/.style={black!25,on/.style={inner sep=-.4pt,circle}}]{
    \node [node, minimum size=3mm] at (1,0) (w) {};
    \node [node, minimum size=3mm,fill=black] at (0,0) (b) {};
    \draw (b) edge [bend left=20,#1] node [on] {\scriptsize#2} (w);
    \draw (w) edge [bend left=20,#3] node [on] {\scriptsize#4} (b);
  }}%
}

\begin{document}
\maketitle

\begin{abstract}\abstractsmall
  Descriptive complexity theory aims at inferring a problem's
  computational complexity from the syntactic complexity of its
  description. A cornerstone of this theory is Fagin's Theorem,
  by which a graph property is expressible
  in \emph{existential second-order logic} (\textsc{eso} logic) if,
  and only if, it is in~$\Class{NP}$. A natural question, from the
  theory's point of view, is which syntactic fragments
  of \textsc{eso} logic also still characterize $\Class{NP}$. Research on this
  question has culminated in a  dichotomy result by Gottlob, Kolaitis,
  and Schwentick: for each possible quantifier prefix of an
  \textsc{eso} formula, the resulting prefix class either contains an
  $\Class{NP}$-complete problem or is contained in~$\Class
  P$. However, the exact complexity of the prefix classes
  inside~$\Class P$ remained  elusive. In the present paper, we clear up the 
  picture by showing that for each prefix class of \textsc{eso} logic, its
  reduction closure under first-order reductions is either
  $\Class{FO}$, $\Class L$, $\Class{NL}$, or $\Class{NP}$. For
  undirected self-loop-free graphs two containment results are
  especially challenging to prove: containment in~$\Class L$ for the
  prefix $\exists R_1\cdots  
  \exists R_n \forall x \exists y$ and containment in~$\Class{FO}$ for
  the prefix $\exists M \forall x \exists y$ for monadic~$M$. The complex
  argument by Gottlob, Kolaitis, and Schwentick concerning polynomial
  time needs to be carefully reexamined and either combined with
  the logspace version of Courcelle's Theorem or directly improved to
  first-order computations. A different challenge is posed by formulas with the
  prefix $\exists M \forall x\forall y$, which we show to express 
  special constraint satisfaction problems that lie
  in~$\Class L$. 
\end{abstract}

\section{Introduction}

Fagin's Theorem \cite{Fagin1974} establishes a tight connection between
complexity theory and finite model theory: A language lies in
$\Class{NP}$ if, and only if, it is the set of all finite models
(coded appropriately as words) of some formula in \emph{existential
  second-order logic} (\textsc{eso} logic). This machine-independent
characterization of a major complexity class sparked the research area
of descriptive  complexity theory, which strives to characterize the
computational complexity of languages by the syntactic structure of
the formulas that can be used to describe them. Nowadays, syntactic logical
characterizations have been found for all major complexity classes, see
\cite{Immerman1998} for an overview, although some syntactic extras
(like numerical predicates) are often needed for technical reasons. 

When looking at subclasses of $\Class{NP}$ like $\Class P$,
$\Class{NL}$, $\Class L$, or $\Class{NC}^1$, one might hope that syntactic
restrictions of \textsc{eso} logic can be used to characterize them;
and the most natural way of restricting \textsc{eso} formulas is to
limit the number and types of quantifiers used. All \textsc{eso}
formulas can be rewritten in prenex normal form as $\exists R_1 \cdots
\exists R_r \forall x_1 
\exists x_2 \cdots \forall x_{n-1} \exists x_n\, \psi$, where the $R_i$
are second-order variables, the $x_i$ are first-order variables, and
$\psi$ is quan\-ti\-fier-free. Formulas like 
$\phi_{\text{3-colorable}} = \exists R \exists G \exists B \forall x
\forall y \bigl(R(x) \lor G(x) \lor B(x) \land (E(x,y) \to\penalty-500
\lnot (R(x) \land\nobreak R(y)) \land  \lnot (G(x) \land G(y)) \land \lnot
(B(x) \land B(y)))\bigr)$, which describes the $\Class{NP}$-complete
problem \Lang{3-color\-able}, show that we  do not need the 
full power of \textsc{eso} logic to capture 
$\Class{NP}$-complete problems: the prefix $\exists R \exists G \exists B \forall x
\forall y$ suffices. However, do formulas of the form,
say, $\exists R \forall x \exists y\, \psi$ also capture all
of $\Class{NP}$; or do they characterize exactly, say,~$\Class P$?
This question lies at the heart of a detailed study by Gottlob,
Kolaitis, and Schwentick \cite{GottlobKS2004} entitled \emph{Existential
  Second-Order Logic Over Graphs: Charting the Tractability Frontier,}
where the following dichotomy is shown: For each possible syntactic
restriction of the quantifier block of \textsc{eso} formulas, the
resulting \emph{prefix class} either contains an $\Class{NP}$-complete
problem or is contained in $\Class P$. For instance, it is shown there
that all graph problems expressible by formulas of the form $\exists R
\forall x \forall y\, \psi$ lie in $\Class P$, while some problems
expressible by formulas of the form  $\exists R \forall x \forall y
\forall z\, \psi$ are $\Class{NP}$-complete.
The dichotomy does not, however, settle the question of whether all of
$\Class P$ -- or at least some interesting subclass thereof like
logarithmic space ($\Class L$) or nondeterministic logarithmic space
($\Class{NL}$) -- is described by one of the logical 
fragments.

\subsection{Contributions of This Paper}

One cannot really hope to show that the prefix class of, say, the
quantifier prefix $\exists R \forall x \forall y$ is  
equal to~$\Class P$ since $\Class P \neq \Class{NP}$ would follow:  
This syntactically severely restricted 
prefix class can be shown \cite[Proposition 10.6]{EiterGG2000} to be
contained in $\Class{NTIME}(n^k)$ for some constant~$k$ and is thus
provably different from~$\Class{NP}$ by the time hierarchy
theorem. The best one can try to prove are statements like ``this
prefix class is contained in $\Class{P}$ and contains a problem
complete for $\Class{P}$'' or, phrased more succinctly, ``the
reduction closure of this prefix class is~$\Class{P}$.'' Our main
result, Theorem~\ref{thm:main}, consists of such
statements: \emph{For each possible \textsc{eso} prefix class, its
  reduction closure under first-order reductions is either
  $\Class{FO}$, $\Class L$, $\Class{NL}$, or $\Class{NP}$.} In
particular, no prefix class yields $\Class P$ as its reduction closure
(unless, of course,
$\Class P = \Class{NP}$ or $\Class{NL} = \Class P$).  

It makes a difference which vocabulary we are allowed to use in our
formulas and which logical structures we are interested in: Results
depend on whether we consider arbitrary graphs, 
undirected graphs, undirected graphs without self-loops, or just strings. (In
this paper, all considered graphs are finite.) The case of
strings has been addressed and settled in
\cite{EiterGG2000}. In the 
present paper we consider the same three cases as in
\cite{GottlobKS2004}: In our vocabulary, we always have just a single
binary relational symbol~($E$), so all models of formulas are
graphs. We then differentiate between directed graphs, undirected
graphs, and undirected graphs without self-loops (which we call
\emph{basic graphs} for brevity). Note that allowing self-loops, whose
presence at a vertex~$x$ can be tested with the formula $E(x,x)$, is
equivalent to considering basic graphs together with an
additional monadic input predicate. 

To describe the syntactic fragments of \textsc{eso} logic easily and
succinctly, we use the notation of~\cite{GottlobKS2004}: The uppercase
letter~$E$ denotes the presence of an existential second-order
quantifier, an optional index as in $E_2$ denotes the arity of the quantifier,
and the  lowercase letters $a$ and~$e$ denote universal and existential
first-order quantifiers, respectively. The prefix type of the formula
$\phi_{\text{3-colorable}}$ mentioned earlier is $EEEaa$ (or even
$E_1E_1E_1 aa$ since the predicates are monadic) and we say 
that \emph{$\phi_{\text{\normalfont 3-colorable}}$ has prefix type $EEEaa$} (and
also $E_1E_1E_1 aa$). We will use regular expressions over the alphabet 
$\{a,e,E,E_1,E_2,E_3,\dots\}$ to denote patterns of prefix types such
as $E^*aa$ for ``any number of existential second-order quantifiers
followed by exactly two universal first-order quantifiers.'' 
To define the three kinds of prefix classes that we are interested
in, for a formula~$\phi$ let $\Models{directed}(\phi) = \{G 
\mid G$ is a directed graph and $G \models \phi\}$,
$\Models{undirected}(\phi) = \{G \mid G$ is an undirected graph and $G
\models \phi\}$, and $\Models{basic}(\phi) = \{G \mid G$ is a
basic graph and $G \models \phi\}$. For instance,
$\Models{basic}(\phi_{\text{3-colorable}}) = \Lang{3-colorable}$
(ignoring coding issues). Next, for a 
prefix type pattern~$P$, let $\FaginDef{directed} (P) =
\{\Models{directed}(\phi) \mid \phi$ has 
a prefix type in~$P\}$ and define $\FaginDef{undirected}(P)$ and
$\FaginDef{basic}(P)$ similarly for undirected and basic graphs. ``$\Class{FD}$'' stands for
``Fagin-definable'' and Fagin's Theorem can be stated succinctly as
$\FaginDef{strings} (E^*(ae)^*) = \Class{NP}$.

As stated earlier, in the context of syntactic fragments of
\textsc{eso} logic it makes sense to consider reduction closures of
prefix classes rather than the prefix classes themselves. It will
not matter much which particular kind of reductions we use, as long
as they are weak enough. All our reductions will be \emph{first-order
  reductions}~\cite{Immerman1998}, which are first-order
queries with access to the bit predicate or, equivalently, functions
computable by a logarithmic-time-uniform constant-depth circuit
family.\footnote{\specialsmall As a technicality, since we use
  first-order reductions with access to the bit predicate, by
  $\Class{FO}$ we refer to ``first-order logic with access to the bit
  predicate,'' which is the same as logarithmic-time-uniform $\Class{AC}^0$.}
Let us write $A \le_{\mathrm{fo}} B$ if $A$ can be reduced to $B$ using
first-\penalty0order reductions. Let us write $\FaginDefbar{directed}(P) = \{A \mid A
\le_{\mathrm{fo}} B \in \FaginDef{directed}(P)\}$ for the reduction closure
of $\FaginDef{directed}(P)$ and define $\FaginDefbar{undirected}(P)$ and $\FaginDefbar{basic}(P)$
similarly. 

\begin{theorem}[Main Result]\label{thm:main}
  The following table completely classifies all prefix classes of
  \textsc{eso} logic over basic graphs (upper part) and
  undirected and directed graphs (lower part):\footnote{The
  ``interesting'' prefixes, where the complexity classes differ
  between the two parts, are highlighted.}
  \medskip

  \noindent\setlength{\tabcolsep}{0pt}\small%
  \begin{tabular}{lll}
    If $P$ is at least one of \dots & and at most one  of \dots,\kern1em & then \\\hline\\[-.75em]
    -- & {$(ae)^*$}, {$E^* e^* a$}, \hilight{$E_1ae$} & $\FaginDefbar{basic}(P)=\Class{FO}$  \\
    \hilight{$E_1E_1ae$}, \hilight{$E_2ae$} & \hilight{$E^* ae$} & $\FaginDefbar{basic}(P)=\Class{L}$ \\
    \hilight{$E_1aa$} & \hilight{$Eaa$} & $\FaginDefbar{basic}(P)=\Class{L}$ \\
    {$E_1eaa$} & {$E_1 e^* aa$} & $\FaginDefbar{basic}(P)=\Class{NL}$ \\
    {$E_1aaa$}, {$E_1E_1aa$}, {$E_2 eaa$}, $E_1eae$,\kern1em &   &   \\
    {$E_1aee$}, {$E_1aea$}, {$E_1 aae$} & {$E^*(ae)^*$} & $\FaginDefbar{basic}(P)=\Class{NP}$ \\[1em]
    -- & {$(ae)^*$}, {$E^* e^* a$} & $\FaginDefbar{undirected}(P) = \FaginDefbar{directed}(P) =\Class{FO}$ \\
    \hilight{$E_1aa$} & {$E_1e^*aa$}, \hilight{$Eaa$} & $\FaginDefbar{undirected}(P) = \FaginDefbar{directed}(P)=\Class{NL}$ \\
    {$E_1aaa$}, {$E_1E_1aa$}, {$E_2eaa$}, \hilight{$E_1ae$} & {$E^*(ae)^*$} & $\FaginDefbar{undirected}(P) = \FaginDefbar{directed}(P)=\Class{NP}$  
  \end{tabular}
\end{theorem}

Note that we always have $\FaginDefbar{undirected}(P) = \FaginDefbar{directed}(P)$, which
is not trivial, especially for the prefix $E_1 aa$: On undirected
graphs, using only two universally quantified variables, it seems
difficult to express ``non-symmetric'' properties, suggesting
$\FaginDef{undirected}(E_1aa) \subseteq \Class L$. However, using a gadget
construction, we will show that $\FaginDef{undirected}(E_1aa)$ contains
an $\Class{NL}$-complete problem.

As an application of the theorem, let us use it to prove
$\Lang{even-cycle} \in \Class L$, which is the problem of detecting
the presence of a cycle\footnote{A cycle in an undirected graph must,
  of course, have length at least $3$ and consist of distinct
  vertices.} of even length in basic graphs~$B$. The complexity of
this problem has been researched for a long time, see
\cite{HemaspaandraST2004} for a discussion and variants. The idea is
to consider the following \textsc{eso} formulas: 
\begin{align} \textstyle
  \phi_m =  \exists
  C_1 \cdots \exists C_m \forall x \exists y \Big( 
  E(x,y) \land \bigvee_{i=1}^m \big( C_i(x) \land C_{(i \bmod m)+1}(y) \land \bigwedge_{j\neq i}
  \neg C_j(x)\big)\Big).
  \label{eq:phim}
\end{align}
They ``describe'' the following situation: The basic graph can
be colored with $m$ different colors so that each vertex $x$ is
connected to a ``next'' vertex $y$ with the ``next'' color (with color $C_1$
following $C_m$). For $m > 2$, it is not hard to see that $B \models
\phi_m$ if, and only if, every connected component of~$B$ contains a 
cycle whose length is a multiple of~$m$. Since $\phi_m$ has quantifier
prefix $E^*ae$ and the graphs are basic, the second row concerning
basic graphs in Theorem~\ref{thm:main} tells
us that $B \models \phi_m$  can be decided in logarithmic space. The
following algorithm now shows $\Lang{even-cycle} \in \Class L$:
In a basic input graph~$B$, replace all edges by length-$2$ paths, then test
whether $C \models \phi_4$ holds for some connected component $C$ of~$B$.

\subsection{Technical Contributions}

The proofs of the statements $\FaginDef{basic}(E^*ae)
\subseteq \Class L$ and  $\FaginDef{basic}(E_1 ae) \subseteq \Class{FO}$  
require a sophisticated technical machinery. In both cases, our proofs
follow the ideas of a 35-page proof of $\FaginDef{basic}(E^*ae)
\subseteq \Class P$ in~\cite{GottlobKS2004}. The central observation
concerning the first statement is that the
\emph{algorithmically} most challenging part in the proof of
\cite{GottlobKS2004} is the application of Courcelle's Theorem~\cite{Courcelle1990a} to
graphs of bounded tree width. It has been shown in
\cite{ElberfeldJT2010} that there is a logspace version of Courcelle's
Theorem, which will allow us to lower the complexity from $\Class P$
to~$\Class L$ when the input graphs have bounded tree
width. For graphs of unbounded tree width, we will explain how the
other polynomial time procedures from the proof of  
\cite{GottlobKS2004} can be reimplemented in logarithmic space. 

To prove $\FaginDef{basic}(E_1 ae) \subseteq \Class{FO}$, we need to lower
the complexity of the involved algorithms further. The idea is to again follow
the ideas from \cite{GottlobKS2004} for $E_1^* ae$. When there is just a single
monadic predicate, certain algorithmic aspects of the proof can be
simplified so severely that they can 
actually be expressed in first-order logic. Note, however, that
already a second monadic predicate or 
a single binary predicate makes the complexity jump up to~$\Class L$,
that is, $\FaginDefbar{basic}(E_1E_1ae) = \FaginDefbar{basic}(E_2ae) = \Class L$.

Concerning the remaining claims from Theorem~\ref{thm:main} that are
not already proved in \cite{GottlobKS2004}, two cases are noteworthy:
Proving that $\FaginDef{basic}(E_1eaa)$ contains an $\Class{NL}$-complete
problem turns out to require a nontrivial gadget construction. Proving
$\FaginDef{basic}(E_1aa) \subseteq \Class L$ requires a reformulation of
the problems in $\FaginDef{basic}(E_1aa)$ as special constraint
satisfaction problems and showing that these lie in~$\Class L$.

\subsection{Related Work}

The study of the expressive power of syntactic fragments of logics
dates back decades; the decidability of prefix classes of first-order
logic, for instance, has been solved completely in a long sequence of
papers, see \cite{BoegerGG1997} for an overview. Interestingly, the
first-order Ackermann prefix class~$ae$ plays a key role in that
context and both $E_1ae$ and $E^* ae$ turn out to be the most
complicated cases in the context of the present paper, too. The
expressive power of monadic second-order logic (\textsc{mso} logic)
has also received a lot of attention, for instance in
\cite{Buechi1960,Courcelle1990a,ElberfeldGT2012}, but emphasis has been
on restricted structures rather than on syntactic fragments.  

Concerning syntactic fragments of \textsc{eso} logic, the two papers
most closely related to the present paper are \cite{EiterGG2000} by
Eiter, Gottlob, and Gurevich and \cite{GottlobKS2004} by Gottlob,
Kolaitis, and Schwentick. In the first paper, a similar kind of
classification is 
presented as in the present paper, only over \emph{strings} rather
than \emph{graphs}. It is shown there that for all prefix patterns $P$
the class $\FaginDef{strings}(P)$ is
either equal to $\Class{NP}$; is not equal to $\Class{NP}$ but
contains an $\Class{NP}$-complete problem; is equal to $\Class{REG}$;
or is a subclass of $\Class{FO}$. Interestingly, two classes of
special interest are $\FaginDef{strings}(E_1^*ae)$ and 
$\FaginDef{strings}(E_1^*aa)$, both of which are the minimal classes equal
to $\Class{REG}$ (by the results of Büchi~\cite{Buechi1960}). In
comparison, by the results of the present paper $\FaginDefbar{basic}(E_1^*ae) =
\FaginDefbar{basic}(E_1E_1ae) = \Class L$, while $\FaginDefbar{basic}(E_1ae) = \Class{FO}$,
and $\FaginDefbar{basic}(E_1^*aa) = 
\FaginDefbar{basic}(E_1E_1aa) = \Class{NP}$, while $\FaginDefbar{basic}(E_1aa) = \Class{L}$.

The present paper builds on the paper \cite{GottlobKS2004} by
Gottlob, Kolaitis, and Schwentick, which contains many of the upper
and lower bounds from Theorem~\ref{thm:main} for the class
$\Class{NP}$ as well as most of the 
combinatorial and graph-theoretic arguments needed to prove
$\FaginDef{basic}(E^*ae) \subseteq\nobreak \Class L$ and  $\FaginDef{basic}(E_1ae)
\subseteq \Class{FO}$. The paper misses, however, the finer
classification provided in our
Theorem~\ref{thm:main} and Remark~5.1 of \cite{GottlobKS2004}
expresses the unclear status of the exact complexity of
$\FaginDef{basic}(E^*ae)$ at the time of writing, which hinges on a
problem called $\Lang{satu}(P)$: ``Note also that for each 
$P$, $\Lang{satu}(P)$ is probably not a $\Class{PTIME}$-complete
set. [\dots] This is due to the check for bounded treewidth, which is
in $\Class{LOGCFL}$ (cf.\ Wanke [1994]) but not known to be in
$\Class{NL}$.'' The complexity of the check for bounded tree width was
settled only later, namely in a paper by Elberfeld, Jakoby, and the
author \cite{ElberfeldJT2010}, and shown to lie in $\Class L$. This
does not mean, however, that the proof of \cite{GottlobKS2004}
immediately yields $\FaginDef{basic}(E^*ae) \subseteq \Class L$ since the
application of Courcelle's Theorem is but one of several subprocedures
in the proof and since a generalization of tree width rather than
normal tree width is used.

\subsection{Organization of This Paper}

To prove Theorem~\ref{thm:main}, we need to prove the lower bounds
implicit in the first column of the theorem's table and the upper
bounds implicit in the second column. The lower bounds are proved
in Section~\ref{section:lower} by presenting reductions from 
complete problems for $\Class L$, $\Class{NL}$, or~$\Class{NP}$. The upper bounds are proved in
Section~\ref{section:upper}, where we prove, in order,
$\FaginDef{basic}(Eaa) \subseteq \Class L$, $\FaginDef{basic}(E^*ae) \subseteq
\Class L$, and $\FaginDef{basic}(E_1ae) \subseteq \Class{FO}$ using
arguments drawn from different areas.

\section{Lower Bounds: Hardness for L and NL}
\label{section:lower}

For each of the prefix patterns listed in the first column of
the table in Theorem~\ref{thm:main} we now show that their prefix
classes contain problems that are hard for $\Class L$, $\Class{NL}$,
or~$\Class{NP}$. The problems from which we reduce are  
listed in Table~\ref{tab:lower}. As can be seen, we only need to prove
new results for a minority of the classes since the $\Class{NP}$ cases
have already been settled in~\cite{GottlobKS2004}.

\begin{table}[ht]
  \caption{The lower bounds in Theorem~\ref{thm:main} are proved by
    showing that the problems in this table, which are complete for
    the classes in the claims, are either expressible in
    the fragment or are at least reducible to a problem
    expressible in the fragment. The problem $\Lang{unreach}$ asks
    whether there is \emph{no} path from $s$ to $t$ in a directed
    graph. The problems $A_2$ and $A_3$ are explained
    below.}\label{tab:lower}  
  \medskip
  
  \small%
  \begin{tabular}{lp{2cm}p{5cm}ll}
    \emph{Claim}&  & \emph{Hard problem} & \emph{Proved where} \\ \hline\\[-.9em]
    \emph{\rlap{Lower bounds for basic graphs}}\\
    $\FaginDefbar{basic}(E_1 E_1 ae) $ & $\supseteq \Class{L}$ & $A_3$ & Lemma~\ref{lem:lower-e1e1ae}\\    
    $\FaginDefbar{basic}(E_2 ae) $ & $\supseteq \Class{L}$ & $A_2$ & Lemma~\ref{lem:lower-e2ae}\\    
    $\FaginDefbar{basic}(E_1 aa) $ & $\supseteq \Class{L}$ & $\Lang{2-colorable}$ & \cite[Remark 3.1]{GottlobKS2004}\\
    $\FaginDefbar{basic}(E_1 eaa) $ & $\supseteq \Class{NL}$ & $\Lang{unreach}$ & Lemma~\ref{lem:lower-e1aaun}\\
    $\FaginDefbar{basic}(E_1 aaa) $ & $\supseteq \Class{NP}$ & $\Lang{positive-one-in-three-3sat}$ & \cite[Theorem 2.2]{GottlobKS2004}\\
    $\FaginDefbar{basic}(E_1E_1 aa) $ & $\supseteq \Class{NP}$ & $ \Lang{3-colorable}$ & \cite[Theorem 2.3]{GottlobKS2004}\\
    $\FaginDefbar{basic}(E_2 eaa) $ & $\supseteq \Class{NP}$ & $\Lang{3-colorable}$ & \cite[Theorem 2.4]{GottlobKS2004}\\
    $\FaginDefbar{basic}(E_1 eae) $ & $\supseteq \Class{NP}$ & $\Lang{3sat}$ & \cite[Theorem 2.5]{GottlobKS2004}\\
    $\FaginDefbar{basic}(E_1 aee) $ & $\supseteq \Class{NP}$ & $\Lang{not-all-equal-3sat}$ & \cite[Theorem 2.6]{GottlobKS2004}\\
    $\FaginDefbar{basic}(E_1 aea) $ & $\supseteq \Class{NP}$ & $\Lang{positive-one-in-three-3sat}$ & \cite[Theorem 2.7]{GottlobKS2004}\\
    $\FaginDefbar{basic}(E_1 aae) $ & $\supseteq \Class{NP}$ & $\Lang{positive-one-in-three-3sat}$ & \cite[Theorem 2.8]{GottlobKS2004}\\[.5em]
    \emph{\kern-1pt\rlap{Remaining lower bounds for undirected and,
        thereby, also for directed graphs}}\\
    $\FaginDefbar{undirected}(E_1 aa) $ & $\supseteq \Class{NL}$ & $\Lang{unreach}$ & Lemma~\ref{lem:lower-e1aaun}\\
    $\FaginDefbar{undirected}(E_1 ae) $ & $\supseteq \Class{NP}$ & $\Lang{3sat}$ & \cite[Theorem 2.1]{GottlobKS2004}
  \end{tabular}
\end{table}

The two special languages $A_2$ and $A_3$ in the table are defined as
follows: For $m\ge 2$ let $A_m = \{G \mid 
G$ is an undirected graph in which each connected component contains a
cycle whose length is a multiple of $m\}$. These languages are all 
hard for $\Class L$: In \cite[page 388, remarks for problem
\Lang{ufa}]{CookM1987} it is shown that the reachability problem for
graphs consisting of just two undirected trees is complete for~$\Class
L$. Since $\Class L$ is trivially closed under complement, testing
whether there is \emph{no} path from a vertex $u$ to a vertex $v$ in a
graph consisting of two trees is also complete for~$\Class
L$, which in turn is the same as asking whether $u$ and $v$ lie in
different trees. To reduce this question to $A_m$, attach cycles of
length $2m$ to both $u$ and~$v$. Then all (namely both) components of the 
resulting graph contain a cycle whose length is a multiple of~$m$
if, and only if, $u$ and $v$ lie in different components. (Using a
cycle length of $2m$ rather than $m$ ensures that also for $m=2$ we
attach a proper cycle.)

\begin{lemma}\label{lem:lower-e1e1ae}
  $A_3 \in \FaginDef{basic}(E_1E_1ae)$. 
\end{lemma}


\begin{proof}
  The discussion following the definition of the formula $\phi_3$ from
  equation~\eqref{eq:phim} shows that $\Models{basic}(\phi_3) = A_3$
  holds; but $\phi_3$ has the prefix $E_1E_1E_1ae$ rather than
  $E_1E_1ae$. However, from $\phi_3$ we can easily build an equivalent 
  formula $\phi_2'$ that only uses two monadic quantifiers: Instead of
  using one monadic relation for each of the three colors, we can
  encode three (even four) 
  colors using only two monadic relations: a vertex $x$ has the first
  color if $C_1(x) \land C_2(x)$, it has the second 
  color if $C_1(x) \land \neg C_2(x)$, the third if $\neg C_1(x) \land
  C_2(x)$, and the fourth if $\neg C_1(x) \land \neg C_2(x)$. 
\end{proof}

\begin{lemma}\label{lem:lower-e2ae}
  $A_2 \in \FaginDef{basic}(E_2ae)$. 
\end{lemma}


\begin{proof}
  Let $\phi = \exists F \forall x \exists y \bigl(
  E(x,y) \land F(x,y) \land \neg F(y,x) \land (F(x,x) \leftrightarrow
  \neg F(y,y))\bigr)$. Then $\phi$ has prefix type $E_2ae$ and we claim    
  $A_2 = \Models{basic}(\phi)$. To see this, first assume that all
  components in a basic graph $B$ contain a cycle of even length. For
  a given component, color the vertices on the cycle alternatively
  white and black. For black vertices $x$, let $F(x,x)$ hold, while
  for white vertices $x$, let $\neg F(x,x)$ hold. Direct the cycle in
  some way and let $F(x,y)$ hold for any two consecutive vertices
  $x$ and $y$ (with respect to the orientation). For all vertices
  $x$ on the cycle we can now choose a vertex~$y$ (namely the next
  vertex on the cycle) such that the quantifier-free part of $\phi$ is
  true. To extend the construction to all vertices, repeatedly
  pick a vertex~$x$ not yet colored, but connected by an edge to an
  already colored vertex~$y$. Assign the opposite color of $y$ to $x$, set
  $F(x,x)$ or $\neg F(x,x)$ accordingly, and let $F(x,y)$ hold. The
  relation $F$ constructed in this way will now witness $B \models
  \phi$.

  For the other direction, let a relation $F$ be given that witnesses
  $B \models \phi$ and consider any component of~$B$. The formula
  $\phi$ chooses for each vertex~$x$ a vertex~$y$; let us call this
  vertex~$y$ the \emph{witness} $w(x)$
  of~$x$. Clearly, $\phi$ enforces that there is an edge between $x$
  and $w(x)$ in~$B$. Starting at any vertex $x$ in
  the component under 
  consideration, consider the sequence $x_1 = x$, $x_2 =
  w(x_1)$, $x_3 = w(x_2)$, and so
  on. Trivially, $x_i \neq x_{i+1}$ since there are no self-loops in
  a basic graph, but we also have $x_i \neq x_{i+2}$ since $\phi$
  enforces $\neg F(x_{i+1},x_i)$, namely for $x = x_i$, and also
  $F(x_{i+1},x_{i+2})$, namely for $x = x_{i+1}$. Now, since
  the graph is finite, the sequence $(x_1,x_2,\dots)$ must run into a
  cycle and, as we just saw, this cycle must have length at
  least~$3$. Finally, the cycle must have even length since 
  $F(x_i,x_i) \leftrightarrow \neg F(x_{i+1},x_{i+1})$ holds for all
  vertices $x_i$ on the cycle and, thus, exactly every second vertex on
  the cycle has a self-loop attached to it by~$F$.
\end{proof}

\begin{lemma}\label{lem:lower-e1aaun}
  $\Lang{unreach}$ reduces to a problem in $\FaginDef{basic}(E_1 eaa)$ and
  also to a problem in $\FaginDef{undirected}(E_1 aa)$. 
\end{lemma}

\begin{proof}
  Since undirected graphs with self-loops are essentially the same as
  basic graphs with an additional monadic relation (the self-loops
  allow us to ``mark'' vertices) and since a single existential
  first-order quantifier such as the one in $E_1 eaa$ also in some sense
  allows us to single out a set of vertices (those that are connected
  to it), we temporarily consider the vocabulary $(E^2, S^1)$, instead 
  of our usual vocabulary $(E^2)$. Logical structures are now graphs
  together with a set of vertices (modeled by~$S^1$). 
  Our objective is to reduce $\Lang{unreach}$ to $\Models{basic}(\phi)$
  where $\phi$ is an $(E^2,S^1)$-formula of the form $\exists M \forall x \forall
  y \, \psi$ for monadic $M$ and quantifier-free~$\psi$. Let
  $(G,s,t)$ be the input for the reduction, where $G = (V,E)$ is a
  directed graph and $s,t \in V$. We build a new, basic graph $B = (V_B,E_B)$ and
  a subset $S$ of $B$'s vertices as follows: For each vertex $v \in V$
  there will be four vertices in $V_B$, designated $v$, $\bar v$, $v'$,
  and $\bar v'$. The vertices $v'$ and $\bar v'$ will be called the
  \emph{shadow vertices} of $v$ and $\bar v$. The shadow vertices will
  form the set~$S$. We have the following undirected edges in~$B$, see
  Figure~\ref{fig:reduction} for an example of the construction:
  \begin{enumerate}
  \item 
    For every vertex $v \in V$ there are the two edges $\{v, \bar
    v\} \in E_B$ and   $\{v',\bar v'\} \in E_B$ and also the two edges
    $\{v,v'\} \in E_B$ and $\{\bar v, \bar v'\} \in E_B$.\footnote{Using $\{u,v\}$ to indicate an
      undirected edge between $u$ and~$v$ in a basic graph and, in
      not-so-slight abuse of notation, even writing $\{u,v\} \in E_B$, helps
      in distinguishing these edges from directed edges in~$E$. Formally,
      we mean of course $(u,v) \in E_B$ and $(v,u) \in E_B$; and $E_B
      \subseteq V \times V$ holds.} 
  \item For every edge $(u,v) \in E$ of the graph~$G$, there 
    is an edge $\{u,v'\} \in E_B$.
  \item There are edges $\{\bar s, s'\}\in E_B$ and $\{t,\bar t'\}\in E_B$.
  \end{enumerate}

  \begin{figure}[htb]
    \centering
    \begin{tikzpicture}
        
      \node (G) at (1.75,0.5) [left] {$G$\rlap{$\colon$}};
      
      \node (B) at (1.75,-1) [left] {$B$\rlap{$\colon$}};

      \draw [|->] (G) -- node[auto] {\footnotesize the first-order
        reduction} (B);
      
      \foreach \n [count=\i] in {s,a,b,c,t}
      {
        \node (\n o) [node] at (\i*2.25,0.5) {$\n$};
        
        \node (\n)      [node] at (\i*2.25,-1) {$\n$};
        \node (\n b)    [node] at (\i*2.25+1,-1) {$\bar \n$};
        \node (\n p)    [node] at (\i*2.25,-2) {$\n'$};
        \node (\n pb)   [node] at (\i*2.25+1,-2) {$\bar \n'$};
        
        \draw (\n) -- (\n b) -- (\n pb) -- (\n p) -- (\n);
      }
      
      \draw (sb) -- (sp)  (t) -- (tpb);
      
      \foreach \from/\to/\s in {s/a/-45, a/c/-30, b/t/-30}
      {
        \draw (\from o) edge[bend left=20, ->] (\to o);
        \draw (\from) edge[out=\s,in=150] (\to p);
      }
      \foreach \from/\to in {c/b, t/c}
      {
        \draw (\from o) edge[bend left=20, ->] (\to o);
        \draw (\from) edge[in=45,out=-140] (\to p);
      }
      
      \begin{scope}[on background layer]
        \fill [black!20,rounded corners=3mm]
        ([shift={(-1.3mm,-1.3mm)}]sp.south west) rectangle
        ([shift={(1.3mm,1.3mm)}]tpb.north east);
        
        \node [above left=6mm,yshift=-3mm] at (sp) (S) {$S$};
        \draw [black!20, very thick] (sp) -- (S);
      \end{scope}
    \end{tikzpicture}
    \caption{Example of the reduction from
      Lemma~\ref{lem:lower-e1aaun}. The directed graph $G$ on top is
      reduced to the basic graph at the bottom. The edges from the
      ``squares'' are the edges resulting from the first rule, the
      curved edges result from the  second rule, and the two diagonal
      edges result from the last rule.} 
    \label{fig:reduction}
  \end{figure}
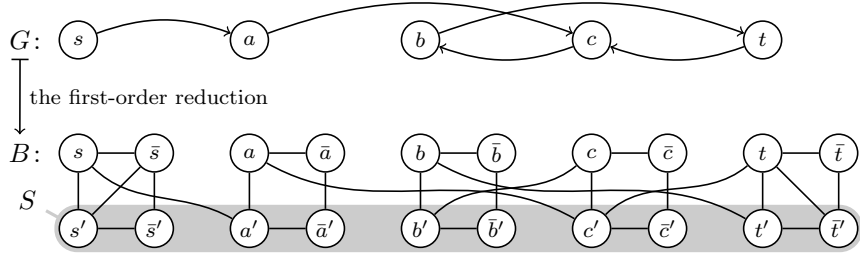

  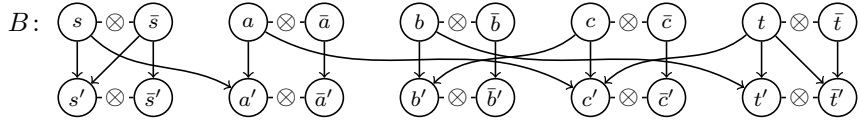
\begin{figure}[htb]
    \centering
    \begin{tikzpicture}
      \node (B) at (1.75,-1) [left] {$B$\rlap{$\colon$}};

      \foreach \n [count=\i] in {s,a,b,c,t}
      {
        \node (\n)      [node] at (\i*2.25,-1) {$\n$};
        \node (\n b)    [node] at (\i*2.25+1,-1) {$\bar \n$};
        \node (\n p)    [node] at (\i*2.25,-2) {$\n'$};
        \node (\n pb)   [node] at (\i*2.25+1,-2) {$\bar \n'$};
        
        \draw (\n) -- node [fill=white,inner sep=1pt] {$\otimes$} (\n b);
        \draw (\n b) edge [->] (\n pb);
        \draw (\n pb) -- node [fill=white,inner sep=1pt] {$\otimes$}(\n p);
        \draw (\n ) edge [->] (\n p);
      }
      
      \draw (sb) edge [->] (sp)  (t) edge [->] (tpb);
      
      \foreach \from/\to/\s in {s/a/-45, a/c/-30, b/t/-30}
      {
        \draw (\from) edge[out=\s,in=150,->] (\to p);
      }
      \foreach \from/\to in {c/b, t/c}
      {
        \draw (\from) edge[in=45,out=-140,->] (\to p);
      }
    \end{tikzpicture}
    \caption{Visualization of the requirements concerning which
      vertices may lie in~$M$  imposed by the formula~$\psi$: For
      edges with label $\otimes$ exactly one end must lie in~$M$ and
      for directed edges, if the tail of the edge lies in~$M$, the
      head must also lie in~$M$.} 
    \label{fig:restr}
  \end{figure}
  
  Let $\phi$ be the following formula:
  \begin{align*}
    \exists M \forall x \forall y \Bigl( E(x,y) \to \bigl(
    & \phantom{{}\land{}} \bigl((\phantom{\neg}S(x) \land
    \phantom{\neg}S(y)) \to (M(x) \leftrightarrow \neg M(y))\bigr)  \\[-2mm]
    & {}\land \bigl((\neg S(x) \land \neg S(y)) \to (M(x)
    \leftrightarrow \neg M(y))\bigr)  \\[-2mm]
    & {}\land \bigl((\neg S(x) \land \phantom{\neg}S(y)) \to (M(x) \to M(y))\bigr)\bigr)\Bigr).
  \end{align*}
  We make some observations concerning how $M$ can be chosen to
  make this formula true: First, we
  only impose restrictions on $M$ when there is an edge 
  between two vertices $x$ and~$y$ in~$B$ (by ``$E(x,y) \to$''). Next, 
  for the edges between vertices inside $S$ (``$S(x) \land S(y)$'') we
  require that exactly one of the two endpoints lies in~$M$. The same
  is true for edges between vertices outside~$S$. Thus, for a vertex
  $v$, we always have either $v \in M$ and $\bar v \notin M$ or
  $v\notin M$ and $\bar v \in M$. Similarly, we always have either $v' \in M$
  and $\bar v' \notin M$ or $v'\notin M$ and $\bar v' \in M$.
  The final restriction (``$\neg S(x) \land S(y)$'') concerns the
  diagonal and curved edges between a vertex and a shadow vertex: Here,
  we require that if $x \in M$ holds, we also have $y \in
  M$. Figure~\ref{fig:restr} visualizes these restrictions for the
  example from Figure~\ref{fig:reduction} by placing an
  $\otimes$-symbol on each edge where exactly one endpoint must be
  in~$M$ and by adding an arrow tip to all edges between a vertex and
  a   shadow vertex.
  
  For any vertex $v \in V$ consider the four vertices $v$, $\bar v$,
  $v'$, and $\bar v'$ in~$B$. Exactly one of $v$ and~$\bar v$ and
  exactly one of $v'$ and~$\bar v'$ must be elements of~$M$. If $v$ is
  an element of~$M$, then so must~$v'$; and if $\bar v$ is an element
  of~$M$, then so must~$\bar v'$. This means 
  that a vertex is an element of $M$ if, and only if, its shadow
  vertex is. Thus, for every vertex $v \in V$ we have $v,v'\in M$ and
  $\bar v,\bar v'\notin M$ or we have $v,v'\notin M$ and $\bar v,\bar
  v'\in M$. Now consider an edge $(x,y) \in E$. If we have $x \in M$,
  then we must 
  also have $y'\in M$ and thus, as we just saw, also $y \in M$. This
  means that when $x \in M$ holds, we also have $z \in M$ for all
  vertices $z$ reachable from $x$ in~$G$. Now, the edge $\{\bar s,
  s'\}$ in $B$ enforces that $s' \in M$ holds (since one of $s$ and
  $\bar s$ will lie in $M$ and the edge from this vertex to $s'$
  enforces that $s'\in M$ holds), which, in turn, enforces $s \in
  M$. The other way round, the edge $\{t, \bar t'\}$ enforces that $t
  \notin M$ holds since, otherwise, we would have both $t'\in M$ and
  also $\bar t'\in M$, which is forbidden.

  Our observations up to now can be summed up as follows: If there is
  some $M$ that makes $\phi$ true, there can be \emph{no} path from
  $s$ to $t$ in $G$ since we must have $s \in M$, $t \notin M$, and
  together with $s$ the set $M$ must contain all vertices reachable
  from~$s$. The other way round, suppose there is no path from $s$
  to~$t$ in $G$. Then the formula $\phi$ is true as the following
  choice for the set~$M$ shows:
  For each vertex $v \in G$, if $v$ is reachable from $s$ in $G$, let
  $v,v'\in M$ and $\bar v, \bar v' \notin M$; otherwise, let
  $v,v'\notin M$ and $\bar v, \bar v' \in M$. Clearly, we now have $s
  \in M$, $t \notin M$, and all requirements of the formula $\phi$ are
  met. This shows that the reduction is correct.
  
  Returning to the original statement of the lemma, we now reduce
  $\Models{basic}(\phi)$ to problems in $\FaginDef{basic}(E_1eaa)$ and
  $\FaginDef{undirected}(E_1aa)$ where there is no $S^1$-predicate any
  longer. For this, let $\psi$ be the quantifier-free part of
  $\phi$. We argue that there are $(E^2)$-formulas
  $\psi'$ and $\psi''$ such that $\Models{basic}(\phi)$ reduces to 
  $\Models{basic}( \exists M \exists z \forall x \forall y \, \psi')$
  and also to
  $\Models{undirected}( \exists M \forall x \forall y \, \psi'')$.  
  
  Switching over to undirected graphs is fairly easy: Construct 
  $\psi''$ from $\psi$ by replacing all occurrences of $S(x)$ by
  $E(x,x)$ and of $S(y)$ by $E(y,y)$. Clearly, we can reduce
  $\Models{basic}(\exists M \forall x \forall y\, \psi)$ to
  $\Models{undirected}(\exists M \forall x \forall y\, \psi'')$ by
  mapping a structure $(V,E,S)$ consisting of a basic graph $B =
  (V,E)$ and a subset $S \subseteq V$ to the undirected graph $(V, E
  \cup \{(x,x) \mid x \in S\})$.
  
  Next, we wish to replace basic graphs with a designated set $S$ by
  basic graphs without such a set, but where a special vertex $z$ can
  be bound by an existential first-order quantifier. Let $\psi'$ be
  obtained from $\psi$ by replacing all occurrences of $S(x)$ and
  $S(y)$ by $E(x,z)$ and $E(y,z)$, respectively, and adding the
  restriction $(x \neq z \land y \neq z) \to \dots$ at the
  beginning, resulting in the following formula $\psi'$:
  \begin{align*}
    (E(x,y) \land x \neq
    z \land y \neq z) \to \bigl(
    & \phantom{{}\land{}} \bigl((\phantom{\neg}E(x,z) \land
    \phantom{\neg}E(y,z)) \to (M(x) \leftrightarrow \neg M(y))\bigr)  \\
    & {}\land \bigl((\neg E(x,z) \land \neg E(y,z)) \to (M(x)
    \leftrightarrow \neg M(y))\bigr)  \\
    & {}\land \bigl((\neg E(x,z) \land \phantom{\neg}E(y,z)) \to (M(x) \to M(y))\bigr)\bigr).    
  \end{align*}
  We claim that 
  $\Models{basic}(\exists M \forall x \forall y\, \psi)$ 
  reduces to $\Models{basic}(\exists M \exists z \forall x \forall
  y\,\psi')$. The reduction would basically like to map a structure
  $(V, E, S)$ to a new basic graph~$B'$ as follows: $B'$ is identical
  to $B= (V,E)$, but
  has a new vertex $z^*$ and edges $\{x,z^*\}$ for all vertices $x \in
  S$. Then if $(V,E,S) \models \exists M \forall x 
  \forall y\, \psi$, we also have $B'\models \exists M \exists z
  \forall x \forall y\,\psi'$ since we can choose $z^*$ in~$\exists
  z$. However, the other direction is not clear: It could 
  happen that $B'\models \exists M \exists z \forall x \forall
  y\,\psi'$, but $z$ is chosen to be some vertex other than~$z^*$ and
  the tests $E(x,z)$, which \emph{should} check whether $S(x)$ used to
  hold in the original graph, test something different.

  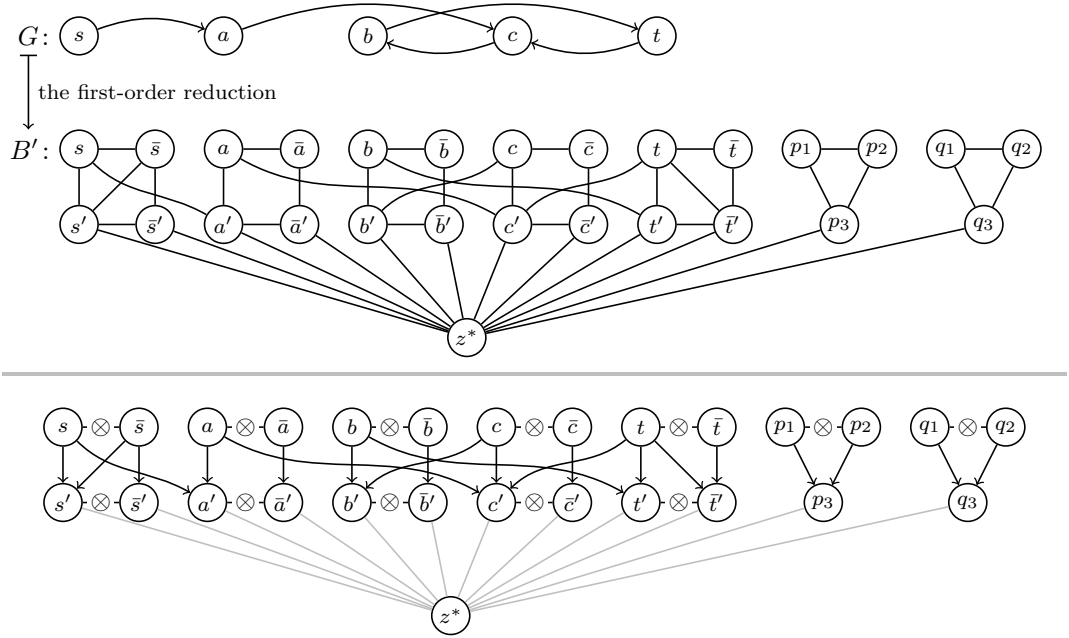
\begin{figure}[htb]
    \centering
        \centering
    \begin{tikzpicture}
        
      \node (G) at (1.5,0.5) [left] {$G$\rlap{$\colon$}};

      \node (B) at (1.5,-1) [left] {\phantom{$G$}\llap{$B'$}{\rlap{$\colon$}}};

      \draw [|->] (G) -- node[auto] {\footnotesize the first-order
        reduction} (B);
      
      \foreach \n [count=\i] in {s,a,b,c,t}
      {
        \node (\n o) [node] at (\i*1.9,0.5) {$\n$};
        
        \node (\n)      [node] at (\i*1.9,-1) {$\n$};
        \node (\n b)    [node] at (\i*1.9+1,-1) {$\bar \n$};
        \node (\n p)    [node] at (\i*1.9,-2) {$\n'$};
        \node (\n pb)   [node] at (\i*1.9+1,-2) {$\bar \n'$};
        
        \draw (\n) -- (\n b) -- (\n pb) -- (\n p) -- (\n);
      }
      
      \draw (sb) -- (sp)  (t) -- (tpb);
      
      \foreach \from/\to/\s in {s/a/-45, a/c/-30, b/t/-30}
      {
        \draw (\from o) edge[bend left=20, ->] (\to o);
        \draw (\from) edge[out=\s,in=150] (\to p);
      }
      \foreach \from/\to in {c/b, t/c}
      {
        \draw (\from o) edge[bend left=20, ->] (\to o);
        \draw (\from) edge[in=45,out=-140] (\to p);
      }
      
      \foreach \n [count=\i] in {p,q}
      {
        \begin{scope}[yshift=-2cm,xshift=\i*1.9cm+9.5cm]
          \node (\n1) [node] at (0,1) {$\n_1$};
          \node (\n2) [node] at (1,1) {$\n_2$};
          \node (\n3) [node] at (.5,0) {$\n_3$}; 
          \draw (\n1) --
          (\n2);
          \draw
          (\n1) edge (\n3) (\n2) edge (\n3);
        \end{scope}
      }

      \node (z) [node] at (7,-3.5) {$z^*$};
      
      \foreach \n in {sp, spb, ap, apb, bp, bpb,
          cp, cpb, tp, tpb, p3, q3} { \draw (z) -- (\n); }
    \end{tikzpicture}
    {
      \medskip
      
      \color{black!25}
      \hrule height1pt
      \bigskip
      
    }
    \begin{tikzpicture}
        
      \foreach \n [count=\i] in {s,a,b,c,t}
      {
        \node (\n)      [node] at (\i*1.9,-1) {$\n$};
        \node (\n b)    [node] at (\i*1.9+1,-1) {$\bar \n$};
        \node (\n p)    [node] at (\i*1.9,-2) {$\n'$};
        \node (\n pb)   [node] at (\i*1.9+1,-2) {$\bar \n'$};
        
        \draw (\n) -- node [fill=white,inner sep=1pt] {$\otimes$} (\n b);
        \draw (\n b) edge [->] (\n pb);
        \draw (\n pb) -- node [fill=white,inner sep=1pt] {$\otimes$}(\n p);
        \draw (\n ) edge [->] (\n p);
      }
      
      \draw (sb)  edge [->] (sp)  (t) edge [->] (tpb);
      
      \foreach \from/\to/\s in {s/a/-45, a/c/-30, b/t/-30}
      {
        \draw (\from) edge[out=\s,in=150,->] (\to p);
      }
      \foreach \from/\to in {c/b, t/c}
      {
        \draw (\from) edge[in=45,out=-140,->] (\to p);
      }
      
      \foreach \n [count=\i] in {p,q}
      {
        \begin{scope}[yshift=-2cm,xshift=\i*1.9cm+9.5cm]
          \node (\n1) [node] at (0,1) {$\n_1$};
          \node (\n2) [node] at (1,1) {$\n_2$};
          \node (\n3) [node] at (.5,0) {$\n_3$}; 
          \draw (\n1) -- node [fill=white,inner sep=1pt] {$\otimes$} (\n2);
          \draw [->] (\n1) edge (\n3) (\n2) edge (\n3);
        \end{scope}
      }

      \node (z) [node] at (7,-3.5) {$z^*$};
      
      \foreach \n in  {sp, spb, ap, apb, bp, bpb,
          cp, cpb, tp, tpb, p3, q3} { \draw[black!25] (z) -- (\n); }
    \end{tikzpicture}
    \caption{
      Example of the reduction from $\Lang{unreach}$ to
      $\Models{basic}(\exists M \exists z \forall x \forall y\, \psi')$ in the
      upper part. The lower part visualizes the conditions imposed
      by the formula $\psi'$ when $z$ is chosen to be~$z^*$ (nothing
      is required concerning the gray lines). Note that the conditions
      on the triangles can easily be satisfied. On the other hand, if
      any vertex other than $z^*$ is chosen, the conditions in at
      least one of the triangles will change to three exclusive ors
      and no solution exists.}
    \label{fig:triangles}
  \end{figure}

  To fix this last problem, we modify the construction of $B'$
  slightly: We add two triangles $p_1$, $p_2$, $p_3$ and $q_1$, $q_2$,
  $q_3$ to $B'$ and additionally the two edges $\{z^*,p_3\}$ and
  $\{z^*,q_3\}$, see Figure~\ref{fig:triangles} for an example. Now,
  if $z$ is chosen as the vertex $z^*$, the edges 
  $\{z^*,p_3\}$ and $\{z^*,q_3\}$ mark $p_3$ and $q_3$ as shadow
  vertices and the conditions imposed by $\psi'$ on the triangle can
  be visualized similarly to Figure~\ref{fig:restr} as shown also in
  Figure~\ref{fig:triangles}. 
  Clearly, the conditions are satisfied when $p_2,p_3,q_2,q_3 \in M$
  and $p_1,q_1 \notin M$.  

  Now suppose that $z$ is not~$z^*$. We claim
  that the formula cannot be true in this case: Whatever vertex we
  choose, the vertices of at least one of 
  the triangles are not connected to the chosen vertex. But, then,
  $\psi'$ enforces that for each edge of the triangle exactly one end
  point lies in~$M$, which is not possible in a triangle, yielding a
  contradiction. 
\end{proof}

\section{Upper Bounds: Containment in FO and L}

\label{section:upper}

The second column of the table in Theorem~\ref{thm:main} lists upper
bounds that we address in the present
section. Table~\ref{tab:upper} shows the order in which we tackle them.

\begin{table}[ht]
  \caption{The upper bounds from Theorem~\ref{thm:main} and where they
    are proved. Missing upper bounds for basic and undirected graphs
    follow from the bounds for directed graphs on the right.}\label{tab:upper}    
  \small%
  \begin{tabular}{lp{1cm}l}
    \rlap{\emph{Claims for basic graphs}}&  & \emph{Proved where} \\ \hline
    $\FaginDef{basic}(E_1ae) $ & $\subseteq \Class{FO}$ & Section~\ref{sec:e1ae} \\
    $\FaginDef{basic}(E^*ae) $ & $\subseteq \Class{L}$ & Section~\ref{sec:estarae} \\
    $\FaginDef{basic}(Eaa) $ & $\subseteq \Class{L}$ &
    Section~\ref{sec:eaa}\\ \\ \\
  \end{tabular}\hfill
  \begin{tabular}{lp{1cm}l}
    \rlap{\emph{Claims for directed graphs}}&  & \emph{Proved where} \\ \hline
    $\FaginDef{directed}((ae)^*) $ & $\subseteq \Class{FO}$ & trivial \\
    $\FaginDef{directed}(E^*e^*a) $ & $\subseteq \Class{FO}$ & \cite[Theorem 3.1]{GottlobKS2004}\\
    $\FaginDef{directed}(E_1e^*aa) $ & $\subseteq \Class{NL}$ & \cite[Theorem 3.2]{GottlobKS2004} \\
    $\FaginDef{directed}(Eaa) $ & $\subseteq \Class{NL}$ & \cite[Theorem 3.4]{GottlobKS2004} \\
    $\FaginDef{directed}(E^*(ae)^*) $ & $\subseteq \Class{NP}$ & Fagin's Theorem
  \end{tabular}

  \vskip-1em
\end{table}

\subsection{\emph{Eaa} Over Basic Graphs:\\ Reformulation as Constraint Satisfaction}
\label{sec:eaa}

Our first upper bound, $\FaginDef{basic}(Eaa) \subseteq \Class{L}$, is
proved in two steps: First, we reformulate the problems in 
$\FaginDef{basic}(Eaa)$ as special constraint satisfaction
problems (\textsc{csp}s) in  Lemma~\ref{lem:ascard}. Second, we show
that these \textsc{csp}s lie in~$\Class L$ in Lemma~\ref{lem:cspinl}.  

It will not be necessary to formally introduce the whole theory of
constraint satisfaction problems since we will only 
encounter one very specialized form of them. Furthermore, our
\textsc{csp}s do not quite fit into the standard framework and major
results on \textsc{csp}s like Schaefer's Theorem~\cite{Schaefer1978}
or the refined version thereof~\cite{AllenderBISV2009} do not settle
the complexity of these special \textsc{csp}s. Nevertheless, we will
need some basic terminology: In a binary \textsc{csp}, we are given a
universe~$U$ and a set of \emph{constraints}, each of which picks a
number of elements from~$U$ and specifies one or more possibilities
concerning which of these elements may lie in a \emph{solution $X
  \subseteq U$}. A \emph{constraint language} specifies the types of
constraints that we are allowed to use. For instance the constraint language for
$\Lang{3sat}$ specifies that constraints (which are clauses) must rule
out one of the eight possibilities concerning which of the elements
(which are the variables) are in~$X$ (are set to $\mathit{true}$). We
need to deviate from this framework in one important way: we require
that there is a constraint for \emph{every} pair of 
distinct elements of~$U$, not just for some of them. Unfortunately,
this deviation inhibits our applying the classification
of the complexity of \textsc{csp}s from~\cite{AllenderBISV2009}; more
precisely, the smallest standard \textsc{csp} 
classes that are able to express the special \textsc{csp}s we are
interested in are known to contain $\Class{NL}$-complete languages --
while we wish to prove containment in~$\Class L$.

For sets $C, D \subseteq \{0,1,2\}$ we define a
\emph{$\{C,D\}$-constraint satisfaction problem $P$ on a universe $U$}
to be a mapping that maps each size-$2$ subset $\{x,y\} \subseteq U$ to
either $C$ or~$D$. A \emph{solution for~$P$} is a subset $X \subseteq
U$ such that for all size-2 subsets $\{x,y\} \subseteq U$ we have
$|\{x,y\} \cap X| \in P(\{x,y\})$. In other words, $P$ fixes for
every pair of two vertices $x$ or $y$ one of two possible constraints
concerning \emph{how many} elements of $\{x,y\}$ may lie in~$X$. Let
$\Lang{csp}\{C,D\} = \{ P \mid P$ is a $\{C,D\}$-\textsc{csp} that has
a solution$\}$. As an example, $\Lang{csp}\bigl\{\{1\},\{0,1,2\}\bigr\}$ is
essentially the same as the problem $\Lang{2-colorable} =
\Lang{bipartite}$ since a $\{1\}$-constraint enforces that exactly one
of two vertices must lie in~$X$ (and, hence, corresponds to an edge),
while a $\{0,1,2\}$-constraint has no effect (and, hence, corresponds
to no edge being present). In Lemma~\ref{lem:cspinl} we show that all
$\Lang{csp}\{C,D\}$ lie in~$\Class L$, which is fortunate since we
reduce to them:

\begin{lemma}\label{lem:ascard}
  For every $Eaa$-formula $\phi$ there are sets $C,D \subseteq
  \{0,1,2\}$ such that the set $\Models{basic}(\phi)$ reduces to
  $\Lang{csp}\{C,D\}$. 
\end{lemma}


\begin{proof}
  We may assume that $\phi$ has the form $\exists M \forall x \forall
  y \, \psi$ with a \emph{monadic} quantifier~$M$ since \cite[Lemma
  3.3]{GottlobKS2004} states that every $Eaa$-formula is equivalent to
  an $E_1aa$-formula. Since the graphs we consider are basic, any
  occurrence of $E(x,x)$ or $E(y,y)$ in $\psi$ can be replaced by just
  $\mathit{false}$. Similarly, $E(y,x)$ can be replaced by 
  $E(x,y)$. Finally, we may assume that $\psi \to x \neq y$ holds as
  well as $\psi(x,y) \leftrightarrow \psi(y,x)$.

  Rewrite $\psi$ equivalently as $x \neq y \to \bigl((E(x,y) \to
  \gamma) \land (\neg E(x,y) \to \delta)\bigr)$ for formulas $\gamma$
  and $\delta$ that are in disjunctive normal form and contain only $M(x)$,
  $M(y)$, $\neg M(x)$, or~$\neg M(y)$ in their terms. Since our graphs are
  basic and the   roles of $x$ and $y$ can be exchanged arbitrarily,
  $\gamma$ and $\delta$ can only make statements about \emph{how many}
  elements of the set $\{x,y\}$ lie in~$M$. For instance, if $\gamma$
  is just $M(x)$, then $\forall x \forall y (E(x,y) \to M(x))$ is
  actually equivalent to $\forall x \forall y (E(x,y) \to (M(x) \land
  M(y)))$ and this imposes the constraint $|\{x,y\} \cap M| = 2$. As
  further examples, $\gamma = (M(x) \land \neg M(y)) \lor (\neg M(x)
  \land M(y))$ imposes the constraint $|\{x,y\} \cap M| = 1$; and
  $\gamma = M(x) \lor M(y)$ imposes the constraint $|\{x,y\}
  \cap M| \in \{1,2\}$. Let $C$ be the cardinality constraints imposed
  by~$\gamma$ and let $D$ be the cardinality constraints imposed
  by~$\delta$ (note that both $C$ and $D$ may be equal to $\emptyset$
  or $\{0,1,2\}$). Then $\Models{basic}(\phi)$ clearly reduces to 
  $\Lang{csp}\{C,D\}$ by mapping each basic graph $B$ to the following
  $\{C,D\}$-\textsc{csp}~$P$: For every edge $\{x,y\}$ of~$B$, let
  $P(\{x,y\}) = C$; and let $P(\{x,y\}) = D$ when 
  there is no edge $\{x,y\}$ in~$B$. 
\end{proof}

\begin{lemma}\label{lem:cspinl}
  Let $C,D \subseteq \{0,1,2\}$. Then $\Lang{csp}\{C,D\} \in \Class L$.
\end{lemma}


\begin{proof}
  Our aim is to explain, for each choice of $C$ and~$D$, how we
  can check in logarithmic space whether a $\{C,D\}$-\textsc{csp} $P$
  has a solution $X \subseteq U$. For a given input $P$, let $B$ be
  the basic graph whose vertex set is~$U$ and which has an edge
  $\{x,y\}$ when $P(\{x,y\}) = C$. Let $\bar B$ be the complement
  graph of~$B$ (exchange edges and non-edges, but do not add
  self-loops). The edges of~$B$ tell us where there are
  ``$C$-constraints'' in~$P$ and the edges of~$\bar B$ where there are
  ``$D$-constraints'' (for $C = D$, the graph $\bar B$ is empty,
  however). We may clearly assume that $B$ has at least three
  vertices.  
  
  We start with some easy observations: If $B$ is the complete graph,
  then there is always a solution if $0 \in C$ (choose $X =
  \emptyset$) or $2 \in C$ (choose $X=U$); there is obviously no
  solution for $C = \emptyset$; and also none for $C = \{1\}$ since
  the graph contains a triangle while $C = \{1\}$ enforces that $B$ must
  be bipartite. We can handle $\bar B$ being the complete graph
  similarly. Thus, we may (1) assume that both $B$ and $\bar B$ contain at
  least one edge. This in turn handles (2) $C = \emptyset$, where there
  can be no solution, and also none for $D = \emptyset$. On the other
  hand, (3) if $0 \in C \cap D$ or $2 \in C \cap D$, there is always a
  solution (namely $X = \emptyset$ or $X = U$). Finally, observe (4)
  that $\Lang{csp}\{C,D\} = \Lang{csp}\bigl\{\{2-c \mid c\in C\},\{2-d\mid
  d\in D\}\bigr\}$ since solutions for \textsc{csp}s of the first kind
  are the complements of solutions for the second kind.

  Let us now go over the cases remaining when $C \neq \emptyset$, $D
  \neq \emptyset$, $0 \notin C \cap D$, and $2\notin C \cap D$:
  \begin{enumerate}
  \item $C = \{0\}$. The remaining choices for $D$ are $\{1\}$,
    $\{2\}$, and $\{1,2\}$ since otherwise by (3) we are done. For
    $D=\{1\}$, a solution can only exist 
    if $\bar B$ is bipartite and $X$ is one of the shores. Both shores
    must be non-empty since $\bar B$ contains an edge by~(1). Since
    shores are independent sets in $\bar B$, the set $X$ must form a
    clique in~$B$. Since no edge of the clique can satisfy the
    constraint $C 
    = \{0\}$, there can be no edges and $|X| = 1$. Thus, all we need to 
    check is whether $\bar B$ is a star, in which case there will be a
    solution. Next, for $D = \{2\}$ there can never be a solution
    since both $B$ and $\bar B$ contain an edge, creating conflicting
    requirements for~$X$. Finally, for $D = \{1,2\}$ if there is any
    solution at all, the set $X = \{v \mid v$ is isolated in $B\}$
    will be such a solution. So, test whether this is indeed the case.
  \item $C = \{2\}$. By observation (4) this case is already settled
    by the previous case.
  \item $C = \{0,2\}$. The only remaining choice for $D$ is
    $\{1\}$. Again, this means that $\bar B$ must be bipartite with
    shores $X$ and $U \setminus X$. Now, if an edge is missing in
    $\bar B$ between a vertex in $X$ and in $U \setminus X$, the
    ``equality constraint'' $C$ cannot be satisfied for this edge
    in~$B$. Thus, $\bar B$ must not only be bipartite, but complete
    bipartite and, then, there is always a solution. All we
    need to test is whether $\bar B$ is complete bipartite (or,
    equivalently, whether $B$ consists of two cliques). Clearly, this
    can be done using even a first-order formula.
  \item $C = \{1\}$. The remaining choices are $D=\{1\}$,
    $D=\{0,1\}$, $D = \{1,2\}$, and $D = \{0,1,2\}$ (the choices
    $\{0\}$, $\{2\}$, and $\{0,2\}$ have already been handled above,
    with the roles of $C$ and $D$ exchanged). For $D = \{1\} = C$
    no solution can exist when the universe has three or more
    elements, which we assume. For $D = \{0,1\}$ the situation is
    similar to the one we had for $C = \{0\}$ and $D = \{1\}$: The constraint $C
    = \{1\}$ enforces that $B$ is bipartite with one shore being~$X$,
    but then $D = \{0,1\}$ enforces that $X$ has size~$1$. So, again,
    we just need to test whether a graph is a star, only this time
    for~$B$. Next, the case $D = \{1,2\}$ is symmetric to $D =
    \{0,1\}$. Finally, for $D = \{0,1,2\}$, the only constraint on~$X$
    is the one given by $C$, which asks whether $B$ is bipartite. This
    test can be done in logarithmic space, however, by Reingold's
    Theorem.
  \item $C = \{1,2\}$. The only remaining choice is $D =
    \{0,1\}$. We claim that there is a solution if, and only if, $B$
    is a \emph{split graph} (a graph whose vertex set can be
    partitioned into two sets $S_{\mathrm{clique}}$ and
    $S_{\mathrm{indep}}$ such that $S_{\mathrm{clique}}$ is a  
    clique and $S_{\mathrm{indep}}$ is an independent set). To see
    this, first note that if $B$ is a split graph, $X =
    S_{\mathrm{clique}}$ satisfies all constraints: Between vertices
    inside $X = S_{\mathrm{clique}}$ there are only $C$-constraint (``pick
    at least one''), between vertices in $U \setminus X =
    S_{\mathrm{indep}}$ there are only $D$-constraint (``pick at most
    one''), and for every pair of vertices where one lies in $X$ and
    the other does not, both a $C$- and a $D$-constraint is always
    satisfied. For the other direction, if $X$ is a solution, then
    there can be no ``at most one'' constraints between the vertices
    in~$X$ and there can be no ``at least one'' constraints between
    the vertices in $U \setminus X$. This shows that $X$ induces a
    clique in $B$ and $U \setminus X$ induces an independent set
    in~$B$. 
    Testing whether $B$ is a split graph can be done using a
    first-order formula since it is known \cite{FoeldesH1977} that a
    graph is a split graph if, and only if, no induced subgraph is
    isomorphic to $2K_2$, $C_4$, or $C_5$.
  \item $C = \{0,1\}$. This is the same as the previous case by
    observation~(4).
  \item $C = \{0,1,2\}$. No untreated choices for $D$ remain. \qedhere
  \end{enumerate}
\end{proof}

\subsection{\emph{E}\boldmath$\mathsf{^*}$\emph{ae} Over Basic Graphs: From P to L}
\label{sec:estarae}

Our objective is to show $\FaginDef{basic}(E^*ae) \subseteq \Class{L}$ in
this section.  More precisely, we only need to show $ \FaginDef{basic}(E^*_1ae) \subseteq
\Class L$ since  \cite[Theorem 4.1]{GottlobKS2004}
states $\FaginDef{basic}(E^*ae) = \FaginDef{basic}(E^*_1ae)$. 

A proof of the weaker claim $\FaginDef{basic}(E_1^*ae) \subseteq
\Class{P}$ is spread over the 35~pages of Sections 4,~5, and~6 of the
paper \cite{GottlobKS2004} by 
Gottlob et al.\ and consists of two kinds of arguments:
Graph-theoretic and algorithmic. Since the graph-theoretic
arguments are independent of complexity-theoretic questions, our main
job is to show how the algorithms described by Gottlob et al.\ can be 
implemented in logarithmic space rather than polynomial time.

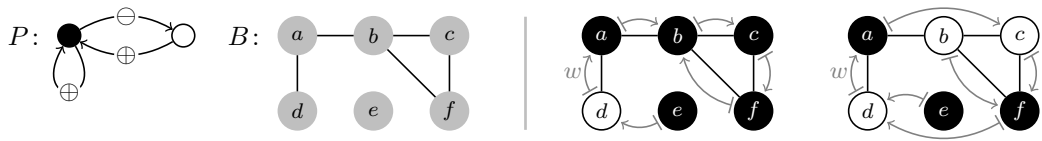
\begin{figure}[htb]
  \centering
  \begin{tikzpicture}
    
    \begin{scope}
      \node [left] at (-.15,0) {$P\colon$};
      
      \node [node, minimum size=3mm] at (1.5,0) (w) {};
      \node [node, minimum size=3mm,fill=black] at (0,0) (b) {};
      
      \draw [overlay] (b) edge [out=-60, in =-120,looseness=15,->] node
      [on,anchor=center] {$\oplus$} (b);
      \draw (b) edge [bend left,->]  node [on] {$\ominus$} (w);
      \draw (w) edge [bend left,->]  node [on] {$\oplus$} (b);
    \end{scope}

    \begin{scope}[xshift=3cm,yshift=-1cm]
      \node [left] at (-.25,1) {$B\colon$};
      \begin{scope}[every node/.style={node,fill=black!25,draw=black!25}]
        \node (a) [at={(0,0)}] {$d$};
        \node (b) [at={(0,1)}] {$a$};
        \node (c) [at={(1,0)}] {$e$};
        \node (d) [at={(1,1)}] {$b$};
        \node (e) [at={(2,0)}] {$f$};
        \node (f) [at={(2,1)}] {$c$};

        \draw (a) -- (b) -- (d) -- (e) -- (f) -- (d);
      \end{scope}
    \end{scope}
    
    \draw [black!25,line width=1pt] (6,.25) -- ++(0,-1.5);
    
    \begin{scope}[xshift=7cm,yshift=-1cm,w/.style={fill=white},b/.style={fill=black,text=white}]
        \node [node] (a) [at={(0,0)}, w] {$d$};
        \node [node] (b) [at={(0,1)}, b] {$a$};
        \node [node] (c) [at={(1,0)}, b] {$e$};
        \node [node] (d) [at={(1,1)}, b] {$b$};
        \node [node] (e) [at={(2,0)}, b] {$f$};
        \node [node] (f) [at={(2,1)}, b] {$c$};

        \draw (a) -- (b) -- (d) -- (e) -- (f) -- (d);
        
        \begin{scope}[every edge/.style={|->,bend left,black!50,draw,shorten >=1pt,shorten <=1pt}]
          \draw (c) edge (a)
                (a) edge node[auto,font=\small,inner sep=1pt]{$w$} (b)
                (b) edge (d)
                (d) edge (f)
                (f) edge (e)
                (e) edge (d);
        \end{scope}
    \end{scope}

    \begin{scope}[xshift=10.5cm,yshift=-1cm,w/.style={fill=white},b/.style={fill=black,text=white}]
        \node [node] (a) [at={(0,0)}, w] {$d$};
        \node [node] (b) [at={(0,1)}, b] {$a$};
        \node [node] (c) [at={(1,0)}, b] {$e$};
        \node [node] (d) [at={(1,1)}, w] {$b$};
        \node [node] (e) [at={(2,0)}, b] {$f$};
        \node [node] (f) [at={(2,1)}, w] {$c$};

        \draw (a) -- (b) -- (d) -- (e) -- (f) -- (d);
        
        \begin{scope}[every edge/.style={|->,bend
            left,black!50,draw,shorten >=1pt,shorten <=1pt}]
          \draw (c) edge[bend right] (a)
                (a) edge node[auto,font=\small,inner sep=1pt]{$w$} (b)
                (b) edge (f)
                (f) edge (e)
                (e) edge (a)
                (d) edge[bend right] (e);
        \end{scope}
    \end{scope}
    
  \end{tikzpicture}
  \caption{Example of a pattern graph $P = (C,A^\oplus, A^\ominus)$ with two ``colors'' black and 
    white (so $C = \{\mathit{black}, \mathit{white}\}$, $A^\oplus =
    \{(\mathit{black},\mathit{black}),
    (\mathit{white},\mathit{black})\}$, and $A^\ominus =
    \{(\mathit{black},\mathit{white})\}$) and an uncolored (``gray'') example
    graph~$B$. We have $B \in \Lang{saturation}(P)$ as shown by two
    examples of legal colorings of $B$ together with witness functions $w$ (in gray).}  
  \label{fig:patternex}
\end{figure}

Similarly to the switch from model checking problems to graphs problems
in the previous section, we also wish to reformulate the model
checking problems $\Models{basic}(\phi)$ for $E_1^*ae$-formulas~$\phi$ in a
graph-theoretic manner. Gottlob et al.\ introduce the notion of \emph{pattern graphs}
for this: A \emph{pattern graph $P = (C,A^\oplus,A^\ominus)$} consists
of a set of \emph{colors~$C$}, a set $A^\oplus \subseteq C\times C$
of \emph{$\oplus$-arcs}, and a set $A^\ominus \subseteq C \times C$ of
\emph{$\ominus$-arcs} ($A^\oplus$ and $A^\ominus$ need not be
disjoint). Given a basic graph $B = (V,E)$, a \emph{coloring 
  of~$G$ with respect to~$P$} is a function $c \colon V \to C$. A
mapping $w \colon V \to V$ is called a \emph{witness function for a
  coloring~$c$} if for all $x \in V$ we have (1)~$x \neq w(x)$,
(2)~if $\{x, w(x)\} \in E$, then $\bigl(c(x), c(w(x))\bigr) \in
A^\oplus$, and (3)~if $\{x, w(x)\} \notin E$, then $\smash{\bigl(c(x),
  c(w(x))\bigr)} \in A^\ominus$.  If there  
exists a coloring together with a witness function for~$B$ with
respect to~$P$, we say that \emph{$B$ can be saturated by~$P$} and the
\emph{saturation problem $\Lang{saturation}(P)$} is the set of all
basic graphs that can be saturated by~$P$, see
Figure~\ref{fig:patternex} for an example.

The intuition behind these definitions is that a
witness function tells us for each $x$ in $\forall x$ which $y$ in
$\exists y$ we must pick to make a formula~$\phi$ of the form
$\exists M_1 \cdots \exists M_n\, \forall x \exists y\, \psi$ true. The pattern
graph encodes the restrictions imposed by $\psi$ and the
monadic predicates~$M_i$:

\begin{fact}[{\cite[Theorem 4.6]{GottlobKS2004}}]\label{fact:pattern}
  For every formula $\phi = \exists M_1 \cdots \exists M_n\, \forall x
  \exists y\, \psi$, where the $M_i$ are monadic and $\psi$ is
  quantifier-free, there is a pattern graph $P$ with $2^n$ vertices 
  such that $\Models{basic}(\phi) =  \Lang{saturation}(P)$.
\end{fact}

Thus, it remains to show $\Lang{saturation}(P) \in \Class L$ for all
pattern graphs~$P$. Towards this aim, for a fixed pattern graph $P$
we devise logspace algorithms that work for larger and larger
classes of basic graphs~$B$, ending with the class of all basic graphs.

\subparagraph*{Graphs of Bounded Tree Width and Special Graphs}
We start by considering only graphs of \emph{bounded tree width,} an
important class of graphs introduced by Robertson and
Seymour in~\cite{RobertsonS1986}: A \emph{tree decomposition} of a
graph~$B$ is a tree $T$ together with a mapping that assigns subsets of $B$'s
vertices (called \emph{bags}) to the nodes of~$T$. The bags must have two properties:
First, for every edge $\{x,y\}$ of $B$ there must be some bag that contains 
both $x$ and $y$. Second, the nodes of $T$ whose bags contain a
given vertex $x$ must be connected in~$T$. The \emph{width} of a
decomposition is the size of its largest bag (minus $1$ for technical
reasons). The \emph{tree width} of~$B$ is the minimal width of
any tree decomposition for it. A class of graphs has \emph{bounded
  tree width} if there is a constant~$c$ such that all graphs in the
class have tree width at most~$c$. From an algorithmic point of view,
many problems that can be solved efficiently on trees can also be
solved efficiently on graphs of bounded tree width. Courcelle's
Theorem turns this into a precise statement:

\begin{fact}[Courcelle's Theorem, \cite{Courcelle1990a}]\label{fact:courcelle}
  For every \textsc{mso}-formula $\phi$ and $t \ge 1$ we have
  \begin{align*}
    \Models{basic}(\phi) \cap 
    \{ G \mid G\text{ has tree width at most }t\} \in \Class{LINTIME}.
  \end{align*}
\end{fact}

Gottlob et al.\ apply this theorem to show that when the input
graphs~$B$ have bounded tree width, we can decide whether $B \in
\Lang{saturation}(P)$ holds in polynomial time: the property $B \in
\Lang{saturation}(P)$ is easily described in \textsc{mso}
logic. We can lower the complexity from ``polynomial time'' to
``logarithmic space'' by using the following logarithmic space version
of Courcelle's Theorem:

\begin{fact}[Logspace Version of Fact~\ref{fact:courcelle}, \cite{ElberfeldJT2010}]\label{fact:ejt}
  For every \textsc{mso}-formula $\phi$ and $t \ge 1$ we have
  \begin{align*}
    \Models{basic}(\phi) \cap 
    \{ G \mid G\text{ has tree width at most }t\} \in \Class{L}.
  \end{align*}
\end{fact}

In their graph-theoretic arguments, Gottlob et al.\ encounter not only
graphs of bounded tree width, but also graphs that they call
\emph{$(k,t)$-special} and 
which are defined as follows: For a basic graph $B = (V,E)$ let us
call two vertices $u$ and~$v$ \emph{equivalent} if for all
$x \in V\setminus \{u,v\}$ we have $\{u,x\} \in E$ if, and only if,
$\{v,x\} \in E$. Observe that this defines an easy-to-check
equivalence relation on the vertices of~$B$ and that each equivalence
class is either a clique or an independent set of~$B$. A graph is
\emph{$(k,t)$-special} if we can remove (up to) $k$ equivalence classes
$A_1$, \dots, $A_k$ from the graph such that the remaining graph
has tree width at most~$t$.

The intuition behind $(k,t)$-special graphs is that equivalent
vertices are ``more or less  
indistinguishable'' and, thus, for a large enough equivalence
class removing some vertices does not change whether the graph can
be saturated or not. Formally, let $B$ be $(k,t)$-special and let
$A_1,\dots,A_k$ be to-be-removed equivalence classes. We obtain an
\emph{$s$-shrink of $B$} by repeatedly removing vertices from those
$A_i$ that have more than $s$ vertices until all of them have at
most $s$ vertices. The proof of Lemma~6.4 in~\cite{GottlobKS2004}
implies the following two facts:

\begin{fact}\label{fact:shrink1}
  For every $k$, $t$, and pattern graph $P$ there is an $s$ such 
  for every $s$-shrink $B'$ of a $(k,t)$-special graph $B$ we have 
  $B \in \Lang{saturation}(P)$ if, and only if, $B' \in
  \Lang{saturation}(P)$.
\end{fact}

\begin{fact}\label{fact:shrinkbtw}
  An $s$-shrink of a $(k,t)$-special graph has tree width at most
  $t+sk$. 
\end{fact}

In Lemmas 6.3 and 6.4 of \cite{GottlobKS2004}, Gottlob et al.\ present
polynomial-time algorithms for testing whether a graph
is $(k,t)$-special and for computing an $s$-shrink when the
test is positive. The following lemma shows that we can reimplement
these algorithms in a space-efficient manner (which the original
algorithms are not):

\begin{lemma}\label{lem:shrink3}
  For every $s$, $k$, and $t$, there is a logspace computable function
  that maps every $(k,t)$-special graph $B$ to an $s$-shrink of~$B$
  (and all other graphs to ``not $(k,t)$-special'').
\end{lemma}


\begin{proof}
  To check whether a basic graph~$B$ is $(k,t)$-special, simply
  iterate over all tuples $(v_1,\dots,v_k)$ of vertices, remove all
  vertices equivalent to any~$v_i$, and test whether the remaining
  graph has tree width at most~$t$ using the logspace algorithm from
  Fact~\ref{fact:ejt}. When a tuple passes the test, for each $v_i$
  remove all but the lexicographically first $s$ vertices that are
  equivalent to~$v_i$ from the graph. What remains is the desired
  shrink. 
\end{proof}

The following lemma sums up the bottom line of the above discussion:

\begin{lemma}\label{lem:kt}
  For every pattern graph $P$ and all $k$ and $t$ we have
  {\begin{align*}
    \Lang{saturation}(P) \cap \penalty0 \{B \mid B\text{ is $(k,t)$-special\/}\} \in
    \Class L.
  \end{align*}}%
\end{lemma}


\begin{proof}
  Let $B$ be a basic input graph. First, use the algorithm from
  Lemma~\ref{lem:shrink3} to (1) test whether $B$ is $(k,t)$-special
  (and if not, reject) and then to (2) compute a shrink $B'$ of~$B$. By
  Fact~\ref{fact:shrink1} we have $B \in \Lang{saturation}(P)$ if, and
  only if, $B' \in \Lang{saturation}(P)$. Thus, it suffices to decide
  the latter membership problem. However, by Fact~\ref{fact:shrinkbtw}
  the graph $B'$ has bounded tree width and, thus, we can use 
  the logspace version of Courcelle's Theorem from Fact~\ref{fact:ejt}
  to decide whether $B'\in \Lang{saturation}(P)$ holds.
\end{proof}

\subparagraph*{Graphs With Self-Saturating Mixed Cycles}
We extend the class of graphs that our logspace machines can
handle to graphs that are not necessarily $(k,t)$-special, but at
least contain a \emph{mixed self-saturating cycle}. A
\emph{self-saturating cycle of a basic graph $B=(V,E)$ with respect
  to a pattern graph $P = (C, A^\oplus, A^\ominus)$} is a sequence
$(v_1,v_2,\dots,v_{n+1})$ of vertices in~$V$ for $n\ge 2$ where the $v_i$ for
$i \in \{1,\dots,n\}$ are all different, $v_{n+1} = v_1$, and we can
assign colors $c \colon \{v_1,\dots,v_n\} \to C$ such that for all $i\in
\{1,\dots,n\}$ we have: if $\{v_i,v_{i+1}\} \in E$, then
$(c(v_i),c(v_{i+1})) \in A^\oplus$; and if $\{v_i,v_{i+1}\} \notin E$,
then $(c(v_i),c(v_{i+1})) \in A^\ominus$. In other words, $B$
restricted to $\{v_1,\dots,v_n\}$ can be saturated with the
``natural'' witness function that ``moves along'' the cycle. The
following is an easy observation concerning self-saturating cycles:

\begin{lemma}\label{lem:selfsatcycle}
  For every $B \in \Lang{saturation}(P)$ there is a
  self-saturating cycle in~$B$~for~$P$.
\end{lemma}


\begin{proof}
  Let $B = (V,E)$ be saturated with respect to $P = (C, A^\oplus,
  A^\ominus)$ via some coloring $c \colon V 
  \to C$ and a witness function $w \colon V \to V$. Starting at any
  vertex~$v$, consider the sequence 
  $v_1 = v$, $v_2 = w(v_1)$, $v_3 = w(v_2)$, \dots, which must clearly
  run into a cycle at some point. Let $(v_i,v_{i+1},\dots,v_j)$ with
  $v_j = v_i$ be this cycle. (For instance, in
  Figure~\ref{fig:patternex} in the first example, starting at~$e$, 
  we run into the cycle $(b,c,f,b)$; and in the second example, starting
  at $e$, we run into the cycle $(d,a,c,f,d)$.) Clearly, the cycle
  $(v_i,v_{i+1},\dots,v_j)$ is self-saturating as demonstrated by
  the coloring~$c$.  
\end{proof}

A self-saturating cycle is \emph{mixed} if for
some $i,j \in \{1,\dots,n\}$ we have $\{v_i,v_{i+1}\} \in E$ and 
$\{v_j,v_{j+1}\}\notin E$, otherwise the cycle is called
\emph{pure}. In Figure~\ref{fig:patternex}, $(b,c,f,b)$ is a pure
self-saturating cycle and $(a, c, f, d, a)$ is a mixed self-saturating
cycle as proved by the two example colorings. Two facts
concerning mixed self-saturating cycles will be important:

\begin{fact}[{\cite[Lemma~6.5]{GottlobKS2004}}]\label{fact:d}
  For every pattern graph $P$ there is a constant~$d$ such that every
  basic graph that has a mixed self-saturating cycle with respect
  to~$P$ also has such a cycle of length at most~$d$.
\end{fact}

\begin{fact}[{\cite[Section~6.3]{GottlobKS2004}}]\label{fact:6.3}
  For each pattern graph $P$ there exist $k$ and $t$ such that $B
  \in \Lang{saturation}(P)$ holds for all graphs $B$ that contain a
  mixed self-saturating cycle but are not $(k,t)$-special. 
\end{fact}

\begin{lemma}\label{lem:self}
  For every pattern graph $P$, we have
  {
    \begin{align*}
      \Lang{saturation}(P) \cap \{B \mid B\text{
        contains a mixed self-saturating cycle}\} \in \Class L.
    \end{align*}%
  }%
\end{lemma}


\begin{proof}
  Let $k$, $t$, and $d$ be the constants from  Facts~\ref{fact:d}
  and~\ref{fact:6.3}. By Fact~\ref{fact:d}, we can decide whether an
  input graph $B$ contains a mixed self-saturating cycle by iterating
  over all possible cycles of 
  maximum length~$d$ and then testing for all possible colorings
  whether a saturation has been found for the cycle. If $B$ fails
  these tests, we can clearly reject.
  
  Otherwise, $B$ has a mixed self-saturating cycle. Test whether $B$
  is $(k,t)$-special using Lemma~\ref{lem:shrink3} and, if so, use
  Lemma~\ref{lem:kt} to decide whether $B \in \Lang{saturation}(P)$
  holds. Finally, if $B$ is not $(k,t)$-special, we can accept by
  Fact~\ref{fact:6.3}.
\end{proof}

\subparagraph*{Arbitrary Basic Graphs}
The last step is to extend our algorithm to graphs that do not contain
mixed self-saturating cycles (and are not 
$(k,t)$-special, but this will no longer be important). Clearly, by
considering the union of the languages from Lemma~\ref{lem:self}
above and Lemma~\ref{lem:pure} below, we see that
$\Lang{saturation}(P) \in \Class L$ holds for all pattern graphs~$P$.

\begin{lemma}\label{lem:pure}
  For every pattern graph $P$, we have
  {
    \begin{align*}
      \Lang{saturation}(P) \cap \{B \mid B\text{ contains no mixed
        self-saturating cycle}\} \in \Class L.
    \end{align*}%
  }%
\end{lemma}


\begin{proof}
  Let $B$ be our input graph. Using Fact~\ref{fact:d} we can first
  rule out (even using a first-order formula) those $B$ containing a mixed 
  self-saturating cycle. Thus, for $B \in \Lang{saturation}(P)$ to
  hold, all self-saturating cycles of~$B$ must be pure (the reverse is
  not true, however: $B$~could have a pure self-saturating cycle that
  cannot be extended to a coloring of the whole graph). In
  \cite{GottlobKS2004}, this situation is addressed in Theorem~5.17,
  which states (reformulated in the terminology 
  of the present paper): \emph{There is a polynomial-time Turing
    machine that decides $\Lang{saturation}(P)$ correctly whenever all
    self-saturating cycles of the input graph~$G$ are pure.} For the
  proof of this statement, the actual algorithm is summarized at the end of 
  \cite[Theorem~5.14]{GottlobKS2004} as follows: ``In fact, the
  computationally relevant actions of the algorithm described in this proof are:
  ---~Computing the complement $G^c$ of $G$ [\dots]. ---~Determining 
  the connected components of $G$ or $G^c$ [\dots]. ---~Checking
  for each component, whether its treewidth is smaller than a
  constant [\dots]. ---~Performing a constant number of further
  [\dots] actions on single components, such as the procedure
  calls $\mathit{satucheck}_P(G)$ or $\mathit{satucheck}'_P(G)$.''
  The omitted parts (``[\dots]'') are statements about the
  \emph{time complexity} of these operations.

  To see that these operations can also be performed in logarithmic 
  space, first note that the complement graph~$G^c$ ($\bar G$~in the 
  notation of this paper) of~$G$ is obtained by simply exchanging
  edges and non-edges (without introducing self-loops, of
  course). Determining the connected components of an undirected graph
  can be done in logarithmic space using Reingold's 
  algorithm. Determining the tree width of a component can be done
  in logarithmic space~\cite{ElberfeldJT2010}. Finally, the procedure calls
  ``$\mathit{satucheck}_P(G)$ or $\mathit{satucheck}'_P(G)$''  
   consist  of checking whether a graph $G$ of bounded
  tree width satisfies a fixed \textsc{mso} formula, which can be done
  in logarithmic space by Fact~\ref{fact:ejt}.
\end{proof}

\subsection{\emph{E}\boldmath$_\mathsf{1}$\emph{ae} Over Basic Graphs: From L to FO}
\label{sec:e1ae}

Our final task for this paper is showing $\FaginDef{basic}(E_1ae)
\subseteq \Class{FO}$.\footnote{\specialsmall In contrast,
  Lemmas~\ref{lem:lower-e1e1ae} and~\ref{lem:lower-e2ae} show that if
  we have \emph{two} monadic 
quantifiers or one \emph{binary} quantifier, the prefix class contains
an $\Class L$-complete problem.} By Fact~\ref{fact:pattern}, it
suffices to show $\Lang{saturation}(P) \in \Class{FO}$ for all pattern
graphs with \emph{two} colors (denoted ``white'' and ``black'' in
the following) and this will be our objective in this
section.\footnote{\specialsmall In contrast, using   three colors we
  can describe $\Class{L}$-complete problems: $\Lang{saturation}(P) =
  A_3$ where $P$ contains a $\oplus$-labeled 3-cycle and $A_3$ is the $\Class
L$-complete language from Table~\ref{tab:lower}.}

In the previous section we proved $\Lang{saturation}(P) \in \Class L$
for all pattern graphs by developing logspace algorithms that worked
for larger and larger classes of graphs. However, this approach is
bound to fail for the class $\Class{FO}$ since properties like 
``the graph is a tree'' (let alone ``the graph is $(k,t)$-special'') are
not expressible in first-order logic. Instead, in this section we show
$\Lang{saturation}(P) \in \Class{FO}$ directly for each possible
pattern graph with two colors.

The simplest case arises when $P = (C, A^\oplus, A^\ominus)$ is
acyclic (meaning that the directed graph $(C, A^\oplus \cup
A^\ominus)$ is acyclic): Lemma~\ref{lem:selfsatcycle} shows that we
then have $\Lang{saturation}(P) = \emptyset$ since 
self-saturating cycles cannot exist for such~$P$. Thus, we only need to
consider pattern graphs $P$ with cycles (self-loops are also cycles,
here). Since $P$ only has two
colors, there are only few ways in which such cycles may
arise. The more cycles there are, the easier 
it will be to color the graph, so we first handle the case that there
are cycles both in $A^\oplus$ and $A^\ominus$, then that there is a cycle
in $A^\oplus$ or in $A^\ominus$, and finally that there is only a
cycle in $A^\oplus \cup A^\ominus$. 

\begin{lemma}\label{lem:twocycles}
  Let $P =  (\{\mathit{black},\mathit{white}\},A^\oplus,A^\ominus)$
  contain cycles both in $A^{\oplus}$ and~$A^\ominus$. Then
  $\Lang{saturation}(P)$ contains all graphs with at least two
  vertices (and is hence in~$\Class{FO}$).
\end{lemma}


\begin{proof}
  Suppose all vertices of~$B$ have degree at least~$1$. Then $B \in
  \Lang{saturation}(P)$ holds for one of two reasons:
  \begin{enumerate}
  \item If there is a self-loop in $A^\oplus$ at one of the colors
    (\patterntwo{any}{}{->}{$\oplus$}{any}{}{any}{} or
    \patterntwo{->}{$\oplus$}{any}{}{any}{}{any}{} where the gray arcs
    can be arbitrary and also be missing) then we can simply color all 
    vertices with the color of the self-loop. The witness function can
    be set to $w(v) = u$ where $u$ is any neighbor of~$v$.
  \item If there is no self-loop in $A^\oplus$, the cycle in
    $A^\oplus$ must be
    \patterntwo{any}{}{any}{}{->}{$\oplus$}{->}{$\oplus$}. We treat
    each connected component $C$ of~$B$ separately. Pick any vertex $c \in
    C$. For each vertex~$v$ of the component, color it white if it has
    an even distance from~$c$, otherwise color it black. Setup the
    witness function $w$ as follows: Map $c$ to any of its
    neighbors. Map each vertex $v$ in the component to one of its
    neighbors that has distance $1$ less from~$c$. Clearly, such a
    neighbor must exist and it will have the opposite color from~$v$.
  \end{enumerate}
  Now suppose that there is a vertex in $B$ that has degree~$0$. Then
  in the complement graph $\bar B$ all vertices have
  an edge to this vertex and, hence, all have degree at
  least~$1$. We can now repeat the above argument, only for a cycle in
  $A^\ominus$ instead of~$A^\oplus$.
\end{proof}

\begin{lemma}\label{lem:onecycle}
  Let $P =  (\{\mathit{black},\mathit{white}\},A^\oplus,A^\ominus)$
  contain a cycle in $A^{\oplus}$ or in~$A^\ominus$. Then
  $\Lang{saturation}(P) \in \Class{FO}$.
\end{lemma}


\begin{proof}
  By possibly switching to complement graphs, we may assume that 
  there is a cycle in $A^\oplus$. We may also assume that there is
  \emph{no} cycle in $A^\ominus$ since, otherwise, we can apply
  Lemma~\ref{lem:twocycles}. As in the proof of that lemma, if in
  the basic input graph $B = (V,E)$ all vertices have degree at
  least~$1$, then $B \in \Lang{saturation}(P)$ holds; so assume that
  there is a vertex of degree~$0$ in~$B$. Then $A^\ominus = 
  \emptyset$ implies $B \notin \Lang{saturation}(P)$ since there cannot be an
  edge between a degree-$0$ vertex and its witness. Similarly,
  if all vertices of~$B$ have degree~$0$, then $B \notin
  \Lang{saturation}(P)$: Since $A^\ominus$ is acyclic, there is no way
  to assign a color to all vertices. So, in the following we may
  assume that the set $S = \{v \mid v$ has degree at least $1$ in
  $B\}$ is neither empty nor all of~$V$ and that $A^\ominus \neq
  \emptyset$. 

  Since $A^\ominus$ neither contains a cycle nor is empty, it can consist only of a single 
  edge: $A^\ominus = \{(\mathit{black},\mathit{white})\}$ or
  $A^\ominus = \{(\mathit{white},\mathit{black})\}$. Because of the
  symmetry of the colors, we only consider the first case. Suppose
  that the color white lies on a cycle in $A^\oplus$ (either because
  of a self-loop at the white color as in
  \patterntwo{->}{$\oplus$}{any}{}{->}{$\ominus$}{any}{}
  or because of a cycle involving
  both colors as in
  \patterntwo{any}{}{any}{}{->}{$\ominus\oplus$}{->}{$\oplus$}). We
  can now color the graph as follows: Color all 
  vertices in $S$ according to the method of
  Lemma~\ref{lem:twocycles} (either all of them are white or we
  alternate between white and black according to the distance to a
  fixed vertex of each component) and setup the witness function~$w$
  on~$S$. Then some vertex $v_0 \in S$ will be colored white
  (typically, many are white, but at least one vertex will be
  white). Color all vertices in $V \setminus S$ black and set the
  witness function to $w(v) = v_0$ for $v \in V \setminus S$. Clearly,
  there will be no edges between $v$ and $v_0$ and, thus, the
  $\ominus$-arc from black to white is saturated.
  
  Now suppose that the color white does not lie in a cycle in
  $A^\oplus$. With most cases ruled out above, the only way
  this can happen is when there is a $\oplus$-self-cycle at black,
  there is the assumed $\ominus$-arc from black to white, and possibly
  an $\oplus$-arc back from white to black:
  \patterntwo{any,draw=none,overlay}{}{->}{$\oplus$}{->}{$\ominus$}{any,draw=none}{}\ \
  or
  \patterntwo{any,draw=none,overlay}{}{->}{$\oplus$}{->}{$\ominus$}{->}{$\oplus$}~. Clearly,
  in the first case, where the backward $\oplus$-arc is missing, $B
  \notin \Lang{saturation}(P)$ holds since 
  the vertices in $S$ \emph{must} be colored black and there is no way
  to then color the vertices in $V \setminus S$. Thus, let us now
  concentrate on the case
  \patterntwo{any,draw=none,overlay}{}{->}{$\oplus$}{->}{$\ominus$}{->}{$\oplus$}~. We
  distinguish three cases:  

  \begin{enumerate}
  \item \emph{$B$ consists of a single edge $\{u,v\}$ plus some
      isolated vertices.} Then we must have $B
    \notin\Lang{saturation}(P)$: We must color all isolated vertices,
    the vertices in $V \setminus S$, black since there cannot be an edge from them to
    their witness in $B$ and $(\mathit{black},\mathit{white})$ is the
    only edge in $A^\ominus$. Then at least one of the two endpoints
    of the single edge in $B$ (say, $u$) must be white, namely the 
    endpoint that is the witness of at least one vertex in $V
    \setminus B$. This enforces that the other endpoint, $v$, is black (since
    $(\mathit{white},\mathit{black}) \in A^\oplus$ is the only edge
    starting at the color white in the pattern graph). Then $v$ cannot
    have a witness: The vertex~$u$ is white, so no edge in $A^\oplus$ can be
    used, nor is any of the other vertices in $V \setminus S$ white,
    so the edge in $A^\ominus$ cannot be used either. 
  \item \emph{$B$ restricted to $S$ is a matching with at least two
      edges.} In this case, pick the first two edges $\{v_1,v_2\} \in
    E$ and $\{v_3,v_4\} \in E$ and color $v_1$ in
    white, $v_2$ in black, $v_3$ in white, and $v_4$ in black. Define
    the witness function $w$ by $w(v_1) = v_2$, $w(v_2) = v_3$,
    $w(v_3) = v_4$, and $w(v_4) = v_1$. Clearly, the coloring and the
    witness function are correct on the vertex set
    $\{v_1,v_2,v_3,v_4\}$. Extend this to a coloring of all vertices
    as follows: All vertices of $S \setminus \{v_1,\dots,v_4\}$ are
    black and their witness is the other end of the edge they are
    attached to, all vertices of $V \setminus S$ are black and their
    witness is~$v_1$ (which is white and there is no edge in $B$
    between vertices in $V \setminus S$ and $v_1 \in S$).
  \item \emph{At least one connected component of~$B$ contains $3$ or
      more vertices.} Let $C$ be such a component. Consider a
    spanning tree~$T$ of $C$ and let $v$ be a leaf of this tree. Color
    $v$ white and all other vertices in the component black. The
    witness of $v$ is its neighbor $u$ in the spanning tree. The
    witness of $u$ is any of its neighbors other than $v$ (such a
    vertex must exist since the spanning tree contains a path of
    length at least~$2$). The witnesses of all other vertices in the
    component is any of their neighbors in the spanning tree. Clearly,
    each vertex of the component is now connected by an edge in~$E$ to
    a black witness as required by $A^\oplus$. Now color all remaining
    vertices of $S$ black, make any of their neighbors in $B$ their
    witnesses, color all vertices of $V \setminus S$ black, and make
    $v$ their witness. As in the previous case, all vertices of $V
    \setminus S$ now have a white witness and there is no edge between
    them and the witness; which is exactly what $A^\ominus$
    requires. \qedhere
  \end{enumerate}
\end{proof}

We are left with the case that the set $A^\oplus \cup
A^\ominus$ contains a cycle, but neither $A^\oplus$ nor~$A^\ominus$
does. This is only possible when $P$ is either
\patterntwomixed{->}{$\oplus$}{->}{$\ominus$} or
\patterntwomixed{->}{$\ominus$}{->}{$\oplus$}. For this special kind
of cycle,
there is an analogue of Fact~\ref{fact:6.3} that does not refer to
$(k,t)$-special graphs: 

\begin{fact}[{\cite[Lemma~6.7]{GottlobKS2004}}]\label{fact:6.7}
  For every pattern graph $P$, we have $B \in \Lang{saturation}(P)$
  for all~$B$ that contain a self-saturating cycle for~$P$ on  which
  $\oplus$- and $\ominus$-arcs alternate.  
\end{fact}

\begin{lemma}\label{lem:mixedcycle}
  Let $P = (\{\mathit{black},\mathit{white}\}, A^\oplus, A^\ominus)$
  contain a cycle in $A^\oplus \cup A^\ominus$, but none in $A^\oplus$
  nor in $A^\ominus$. Then $\Lang{saturation}(P) \in \Class{FO}$.
\end{lemma}


\begin{proof}
  Let $B$ be a basic input graph. We wish to test whether $B$
  contains a mixed self-saturating cycle for~$P$, which must be 
  \patterntwomixed{->}{$\oplus$}{->}{$\ominus$} or
  \patterntwomixed{->}{$\ominus$}{->}{$\oplus$}. By Fact~\ref{fact:d},
  if such a mixed self-saturating cycle exists, there is one of
  length~$d$ for some constant~$d$. (The proof in \cite{GottlobKS2004}
  yields $d = 2^{76}+2$ for our pattern graph; but a direct argument
  shows that $d=4$ suffices, fortunately.) Thus, the following formula
  tells us whether a mixed self-saturating cycle exists in $B$
  for~$P$:
  \begin{align*}
    \exists a \exists b \exists c \exists d \bigl(&E(a,b) \land \neg E(b,c)
    \land E(c, d) \land \neg E(d,a) \land{}\\
    &a\neq b \land b \neq c \land c\neq d \land a \neq c \land b \neq
    d \land a \neq d \bigr).
  \end{align*}
  We claim that this formula \emph{also} tells us whether $B \in
  \Lang{saturation}(P)$ holds: The existence a mixed self-saturating 
  cycle in~$B$ is a necessary condition for $B \in \Lang{saturation}(P)$ by
  Lemma~\ref{lem:selfsatcycle}. It is also a sufficient condition by
  Fact~\ref{fact:6.7} because of the special structure of the only cycle
  in~$P$.
\end{proof}

\section{Conclusion}

In the present paper we have completely classified the first-order
reduction closures of prefix classes of \textsc{eso} logic over
directed, undirected, and basic graphs: each one of them is equal to
one of the standard classes $\Class{FO}$, $\Class L$, $\Class{NL}$, or
$\Class{NP}$. It turned out that the prefix classes for directed and
undirected graphs are always the same, but often differ from the prefix
classes for basic graphs. Especially interesting prefixes that mark
the border between one complexity class and the next are $E_1ae$, 
$E^*ae$, and~$Eaa$.

A natural question that arises is: Can we find a prefix class whose
reduction closure is~$\Class P$? By the results of 
the present paper, this cannot be an \textsc{eso} prefix class, unless
unlikely collapses occur. However, what about prefix classes of
general second-order logic? We may similarly ask whether any class
other than $\Class L$, $\Class{NL}$, and the classes of the polynomial 
hierarchy can be characterized by a prefix class of second-order
logic. 

Together with the results from \cite{EiterGG2000}, we now have a
fairly complete picture of the complexity of all \textsc{eso} prefix
classes over directed graphs, undirected graphs, basic graphs, and
strings. Concerning arbitrary logical structures, Gottlob et al.\
\cite{GottlobKS2004} already point out that their $\Class
P$-$\Class{NP}$-dichotomy for 
directed graphs generalizes to the collection of all finite structures
over any relational vocabulary that contains a relation symbol of
arity at least two; and it is not hard to see that our
Theorem~\ref{thm:main} also generalizes in this way (a closer look at
the $\Class{FO}$ and $\Class{NL}$ upper bounds in \cite{GottlobKS2004}
shows that they hold for arbitrary structures). The complexity of
prefix classes over other special structures is, however, still open,
including those of trees, infinite words, and bipartite graphs.

\bibliography{main}

\end{document}